%% file: MultiAnd.tex
\documentclass[oribibl]{llncs}

\input{header}

\providecommand{\ext}{{ext}}
\DeclareMathOperator{\e}{e}
\DeclareMathOperator*{\Ex}{\mathbb{E}}
\DeclareMathOperator{\AND}{AND}
\DeclareMathOperator{\OR}{OR}

\DeclareMathOperator{\IC}{IC}
\DeclareMathOperator{\CI}{CI}
\DeclareMathOperator{\entropy}{H}  
\DeclareMathOperator{\MI}{I}  
\DeclareMathOperator{\Dvg}{D}  
\DeclareMathOperator{\supp}{supp}
\providecommand{\concav}{\mathrm{concav}}

\spnewtheorem{assumption}{Assumption}{\bfseries}{\itshape}
\spnewtheorem{statement}{Statement}{\bfseries}{\itshape}
\spnewtheorem{myclaim}{Claim}{\bfseries}{\itshape}

\begin{document}

\title{Information complexity of the AND function in the two-party, and multiparty settings}

\author{
Yuval Filmus\inst{1} 
\and Hamed Hatami\thanks{Supported by an NSERC grant.}\inst{2}
\and Yaqiao Li\inst{3}
\and Suzin You\inst{4} 
}

\institute{
Technion --- Israel Institute of Technology, \email{yuvalfi@cs.technion.ac.il}
\and
McGill University, \email{hatami@cs.mcgill.ca}
\and
McGill University, \email{yaqiao.li@mail.mcgill.ca}
\and
McGill University, \email{suzinyou.sy@gmail.com}
}

\maketitle

\begin{abstract}
In a recent breakthrough paper [M. Braverman, A. Garg, D.  Pankratov, and O. Weinstein, From information to exact communication, STOC'13 Proceedings  of  the  2013  ACM  Symposium  on  Theory  of Computing, ACM, New York, 2013, pp. 151--160.] Braverman {\it et al.}\ developed a local characterization for the zero-error information complexity in the two party model, and used it to compute the exact internal and external information complexity of the $2$-bit AND function.

In this article,  we extend their result to the multiparty number-in-hand model by proving that the generalization of their protocol has optimal internal and external information cost for certain distributions. Our proof has new components, and in particular it fixes some minor gaps in the proof of Braverman {\it et al}. 
\end{abstract}

\section{Introduction}

Although  communication complexity has since its birth been witnessing steady and rapid progress, it was not until recently that a focus on an information-theoretic approach resulted in new and deeper understanding of some of the classical problems in the area. This gave birth to a new area of complexity theory called \emph{information complexity}. Recall that communication complexity is concerned with minimizing the amount of communication required for  players who wish to evaluate a function that depends on their private inputs. Information complexity, on the other hand, is concerned with the amount of information that the communicated bits reveal about the inputs of the players to each other, or to an external observer.

One of the important achievements of information complexity is the recent result of~\cite{MR3210776} that  determines the exact asymptotics of the randomized communication complexity of one of the oldest and most studied problems in communication complexity,  set disjointness:
\begin{equation}
\label{eq:DISJ}
\lim_{\epsilon \to 0} \lim_{n\to\infty} \frac{R_\epsilon(\mathrm{DISJ}_n)}{n} \approx 0.4827.
\end{equation}
Here $R_\epsilon(\cdot)$ denotes the randomized communication complexity with an error of at most $\epsilon$ on every input, and $\mathrm{DISJ}_n$ denotes the set disjointness problem. In this problem, Alice and Bob each receive a subset of $\{1,\ldots,n\}$, and their goal is to determine whether their sets are disjoint or not.   Prior to the discovery of these information-theoretic techniques, proving the lower bound $R_\epsilon(\mathrm{DISJ}_n) = \Omega(n)$ had already been a challenging problem, and even Razborov's~\cite{MR1192778} short proof of that fact is intricate and sophisticated.

Note that the set disjointness function is nothing but an $\OR$ of $\AND$ functions. More precisely, for $i=1,\ldots,n$, if $x_i$ is the Boolean variable which represents whether $i$ belongs to Alice's set or not, and   $y_i$ is the corresponding variable for Bob, then $\bigvee_{i=1}^n (x_i \wedge y_i)$ is true if and only if  Alice's input intersects Bob's input. Braverman {\it et al.}~\cite{MR3210776} exploited this fact to prove \eqref{eq:DISJ}. Roughly speaking, they first determined the exact information cost of the  $2$-bit AND function for any underlying distribution $\mu$ on the set of inputs $\{0,1\} \times \{0,1\}$, and then used the fact that amortized communication equals information cost~\cite{MR3265014} to relate this to the communication complexity of the set disjointness problem.  The constant $0.4827$ in (\ref{eq:DISJ}) is indeed the maximum of the information complexity of the $2$-bit AND function over all
measures $\mu$ that assign a zero  mass to $(1,1) \in \{0,1\} \times \{0,1\}$. That is
$$
\max_{\mu:\mu(11)=0} \IC_\mu(\AND) \approx 0.4827,
$$
where $\IC_\mu(\AND)$ denotes the information cost of the $2$-bit $\AND$ function with respect to the distribution $\mu$ with no error (See Definition~\ref{def:infocost} below). These results show the importance of knowing the exact information complexity of simple functions such as the $\AND$ function.

Although obtaining the asymptotics of $R_\epsilon(\mathrm{DISJ}_n)$ from  the information complexity of the $\AND$ function is not straightforward and a formal proof requires overcoming some technical difficulties, the bulk of~\cite{MR3210776} is dedicated to computing the exact information complexity of the $2$-bit AND function.  This rather simple-looking problem had been studied previously by Ma and Ishwar~\cite{MaIshwar2011,MaIshwar2013}, and some of the key ideas of~\cite{MR3210776} originate  from their work. In~\cite{MR3210776} Braverman {\it et al}  introduced a protocol to solve the $\AND$  function, and proved that it has optimal internal and external information cost. Interestingly this protocol is not a conventional communication protocol as it has access to a continuous clock, and the players are allowed to ``buzz'' at randomly chosen times. However, one can approximate it by conventional  communication protocols through dividing the time into
finitely many discrete units. Indeed, it is known~\cite{MR3210776} that no protocol with a bounded number of rounds can have optimal information cost for the $\AND$ function, and hence the infinite number of rounds, implicit in the continuous clock, is essential. We shall refer to this protocol as the \emph{buzzers} protocol.

\subsection{Our contributions}

{\bf Fixing the argument of \cite{MR3210776}:}
In order to show that the \emph{buzzers} protocol has optimal information cost, inspired by the work of Ma and Ishwar~\cite{{MaIshwar2011,MaIshwar2013}}, Braverman {\it et al} came up with a local concavity condition, and showed that if a protocol satisfies this condition, then it has optimal information cost. This condition, roughly speaking, says that it suffices to verify that one does not gain any advantage over the conjectured optimal protocol if one of the players starts by sending a bit $B$. In the original paper~\cite{MR3210776}, it is claimed that it suffices to verify this condition only for signals $B$ that reveal arbitrarily small information about the inputs. As we shall see, however, this is not true, and one can easily construct counter-examples to this statement. 

In Theorem~\ref{thm:LocalCharNoErrorWeak} we prove a variant of the local concavity condition that allows one to consider only signals $B$ with small information leakage, and then apply it to fix the argument in \cite{MR3210776}. We have been informed through private communication that Braverman {\it et al} have also independently fixed the argument in \cite{MR3210776}.

\smallskip
{\bf \noindent Extension of \cite{MR3210776} to the multi-party setting:} We then apply Theorem~\ref{thm:LocalCharNoErrorWeak} to extend the result of~\cite{MR3210776}  to the multiparty number-in-hand model by defining a generalization of the \emph{buzzers} protocol, and then prove in Theorem~\ref{thm:mainMultiparty} that it has optimal internal and external information cost when the underlying distribution satisfies the following assumption:
\begin{assumption}
\label{assumption}
The support of $\mu$ is a subset of $\{\vec{0},\vec{1},\e_1,\ldots,\e_k\}$, where $\e_i$ is the usual $i^\text{th}$ basis vector $(0,\dots,0,1,0,\dots,0)$.
\end{assumption}
Note that in the two-party setting, every distribution satisfies this assumption and thus our results are complete generalizations of the results of~\cite{MR3210776} in the two party setting. The distribution in Assumption~\ref{assumption} arise naturally in the study of the set disjointness problem, and as a result they have been considered previously in~\cite{specialdistribution}.

This extension  is not  straightforward since in \cite{MR3210776}, a large part of the calculations for verifying the local concavity conditions  are carried out by the software Mathematica.  However, in the number-in-hand model, having an arbitrary  number of players, one cannot simply rely on a computer program for those calculations. Instead, first we had to analyze and understand what happens at different stages of the protocol, and once we reduced the problem to sufficiently simple  equations (with a constant number of variables), then we used a computer program to verify them. We believe our proof provides some new insights even for the two-party setting.

\section{Preliminaries}
\subsection{Notation}
We typically denote the random variables by capital letters (e.g $A,B,C,X,Y,\Pi$). For the sake of brevity, we shall write $A_1\ldots A_n$ to denote the random variable $(A_1,\ldots,A_n)$ and \emph{not} the product of the $A_i$'s. We use $[n]$ to denote the set $\{1,\ldots,n\}$, and $\supp(\mu)$ to denote the support of a measure $\mu$. We denote the statistical distance (a.k.a. total variation distance) of two measures $\mu$ and $\nu$ on the sample space $\Omega$ by $|\mu-\nu| := \frac{1}{2}\sum_{a \in \Omega} |\mu(a)-\nu(a)|$.

For every $\epsilon \in [0,1]$, $\entropy(\epsilon) = -\epsilon \log\epsilon - (1-\epsilon) \log(1-\epsilon)$ denotes the binary entropy, where here and throughout the paper $\log(\cdot)$ is in base $2$, and $0 \log 0 = 0$. 

Recall $\Dvg(\mu || \nu)$ means the divergence (a.k.a. relative entropy, or Kullback Leibler distance) between two distributions $\mu$ and $\nu$. Let $X$ and $Y$ be two random variables,  the standard notation $\MI(X,Y)$ means the mutual information between $X$ and $Y$, sometimes we use the notation $\Dvg(X||Y)$ to denote the divergence between the distributions of $X$ and $Y$. For definitions and basic facts regarding divergence and mutual information, see~\cite{cover2012elements}.

\subsection{Communication complexity}
The notion of two-party communication complexity was introduced by Yao~\cite{Yao:1979} in 1979. In this model there are two players (with unlimited computational power), often called Alice and Bob, who wish to collaboratively compute a given function $f\colon \cX \times \cY \to \cZ$. Alice receives an input $x \in \cX$ and Bob receives $y \in \cY$. Neither of them knows the other player's input, and they wish to communicate in accordance with an agreed-upon protocol $\pi$ to  compute $f(x,y)$.  The protocol $\pi$ specifies as a function of (only) the transmitted bits  whether the communication is over, and if not, who sends the next bit. Furthermore $\pi$ specifies what the next bit must be as a function of the transmitted bits, and the input of the player who sends the bit. The \emph{cost} of the protocol is the total number of bits transmitted on the worst case input. The \emph{transcript} $\Pi$ of a protocol $\pi$ is the list of all the transmitted bits during the execution of the protocol.

In the randomized communication model, the players might have access to a shared random string (\emph{public randomness}), and their own private random strings (\emph{private randomness}).  These random strings are independent, but they can have any desired distributions individually. In the randomized model the transcript also includes the public random string  in addition to the transmitted bits. Similar to the case of deterministic protocols, the \emph{cost} is the total number of bits transmitted on the worst case input and random strings. The \emph{average cost} of the protocol is the expected number of bits transmitted on the worst case input.

For a function $f\colon \cX \times \cY \to \cZ$ and a parameter $\epsilon>0$, we denote by $R_\epsilon(f)$ the cost of the best randomized protocol for computing $f$ with probability of error at most $\epsilon$ on \emph{every} input.

\subsection{Information complexity}  \label{sec:IC-definition}

The setting is the same as  in communication complexity, where Alice and Bob (having infinite computational power) wish to mutually compute a function $f\colon \cX \times \cY \to \cZ$. To be able to measure information, we also need to assume that there is a prior distribution $\mu$ on $\cX \times \cY$.

 For the purpose of communication complexity, once we allow public randomness, it makes no difference whether we permit the players to have private random strings or not. This is because the private random strings can be simulated by parts of the public random string. On the other hand, for information complexity, it is crucial to  permit private randomness, and once we allow private randomness, public randomness becomes inessential. Indeed, one of the players can use her private randomness to generate the public random string, and then transmit it to the other player. Although this might have very large communication cost, it has no information cost, as it does not reveal any information about the players' input.

Probably the most natural way to define the information cost of a protocol is to consider the amount of information that is revealed about the inputs $X$ and $Y$ to an external observer who sees the transmitted bits and the public randomness. This is known as the \emph{external information cost} and is formally defined as the mutual information between $XY$ and the transcript of the protocol (recall that the transcript $\Pi_{XY}$ contains the public random string $R$). While this notion is interesting and useful, it turns out there is another way of defining information cost that  has  some very useful properties. This is called the \emph{internal information cost} or just the \emph{information cost} for short, and is equal to  the amount of information that Alice and Bob learn about each other's input from the communication. Note that Bob knows $Y$, the public randomness $R$ and his own private randomness $R_B$, and thus what he learns about $X$ \emph{from the communication} can be measured by $I(X;\Pi|YRR_B)$. Similarly, what Alice learns about $Y$ from the communication can be measured by  $I(Y;\Pi|XRR_A)$ where $R_A$ is Alice's private random string. It is not  difficult to see~\cite{MR2743255} that  conditioning on the public and private randomness does not affect these quantities. In other words $I(X;\Pi|YRR_B)=I(X;\Pi|Y)$ and $I(Y;\Pi|XRR_A)= I(Y;\Pi|X)$. We summarize these in the following definition.

\begin{definition}
\label{def:infocost}
The \emph{external information cost} and the \emph{internal information cost} of a protocol $\pi$ with respect to a distribution $\mu$ on inputs from $\cX \times \cY$ are defined as
$$\IC_\mu^\ext(\pi) = I(\Pi; XY),$$
and
$$\IC_\mu(\pi) = I(\Pi; X|Y)+I(\Pi; Y|X),$$
respectively, where $\Pi=\Pi_{XY}$ is the transcript of the protocol when it is executed on the inputs $XY$.
\end{definition}

We will be interested in certain \emph{communication tasks}. Let $[f,\eps]$ denote the task of computing  the value of $f(x,y)$ correctly with probability at least $1-\eps$ for \emph{every} $(x,y)$. Thus a protocol $\pi$ performs this task  if
\[
\Pr[\pi(x,y) \neq f(x,y)] \le \epsilon, \quad \forall\ (x,y) \in \cX \times \cY.
\]
Given another distribution $\nu$ on $\cX \times \cY$, let $[f, \nu, \epsilon]$ denote the task of computing  the value of $f(x,y)$ correctly with probability at least $1-\eps$ if the input $(x,y)$ is sampled from the distribution $\nu$. A protocol $\pi$ performs this task  if
\[
\Pr_{(x,y) \sim \nu}[\pi(x,y) \neq f(x,y)] \le \epsilon.
\]
Note that a protocol $\pi$ performs $[f, 0]$ if it  computes $f$ correctly on \emph{every} input while performing $[f, \nu,0]$  means computing $f$ correctly on the inputs that belong to the support of $\nu$. 

The \emph{information complexity}  of a communication task $T$ with respect to a measure $\mu$ is defined as
\[ \IC_\mu(T) = \inf_{\pi :\ \pi \text{\  performs\ } T} \IC_\mu(\pi). \]
It is essential here that we use infimum rather than minimum as there are tasks for which there is no protocol that achieves $\IC_\mu(T)$ while there is a sequence of protocols whose information cost converges to $\IC_\mu(T)$.  The \emph{external information complexity}  of a communication task $T$ is defined similarly. We will abbreviate $\IC_\mu(f,\eps)=\IC_\mu([f,\eps])$,  $\IC_\mu(f,\nu,\eps)=\IC_\mu([f,\nu,\eps])$, etc. 
It is important to note that when $\mu$ does not have full support, $\IC_\mu(f,\mu,0)$ can be strictly smaller than  $\IC_\mu(f,0)$. We sometimes also abbreviate $\IC_\mu(f) = \IC_\mu(f,0)$.

\begin{remark}[A warning regarding our notation]
In the literature of information complexity it is common to use ``$\IC_\mu(f,\epsilon)$'' to denote the distributional error case, i.e. what we denote by $\IC_\mu(f,\mu,\epsilon)$. Unfortunately this has become the source of some confusions in the past, as sometimes  ``$\IC_\mu(f,\epsilon)$'' is used to denote both of the distributional error $[f,\mu,\epsilon]$  and the point-wise error $[f,\epsilon]$. To avoid ambiguity we distinguish the two cases by using the different notations $\IC_\mu(f,\mu,\epsilon)$ and $\IC_\mu(f,\epsilon)$.
\end{remark}

\subsection{The continuity of information complexity}   \label{sec:IC-continuity}

The information complexities $\IC_\mu(f,\epsilon)$ and $\IC_\mu(f,\nu,\epsilon)$ are both continuous with respect to $\epsilon$. The following simple lemma from~\cite{braverman2015interactive} proves the continuity for $\epsilon \in (0,1]$. The continuity at $0$ is more complicated and is proven in~\cite{MR3210776}.

\begin{lemma}\cite{braverman2015interactive}
For every $f\colon\mathcal{X} \times \mathcal{Y} \to \mathcal{Z}$, $\epsilon_2 >\epsilon_1>0$ and measures $\mu,\nu$ on $\mathcal{X} \times \mathcal{Y}$, we have
\begin{equation}
\label{eq:continuityEps}
\IC_\mu(f,\nu,\epsilon_1)- \IC_\mu(f,\nu,\epsilon_2) \le (1-\epsilon_1/\epsilon_2) \log |\mathcal{X} \times \mathcal{Y}|,
\end{equation}
\end{lemma}

\begin{proof}
Let $f\colon\mathcal{X} \times \mathcal{Y} \to \mathcal{Z}$, and consider a protocol $\pi$ with information cost $I$, and error $\epsilon_2>0$. Set $\delta= 1-\epsilon_1/\epsilon_2$, and let $\tau$ be the following protocol
\begin{itemize}
 \item With probability $1-\delta$ run $\pi$.
 \item With probability $\delta$ Alice and Bob exchange their inputs and compute $f(x,y)$.
\end{itemize}
The theorem follows as the new protocol has error at most $(1-\delta)\epsilon_2=\epsilon_1$, and information cost at most $I + \delta \log |\mathcal{X}  \times \mathcal{Y}|$.
\end{proof}

\begin{remark}
The same proof implies
$\IC_\mu(f,\epsilon_1)- \IC_\mu(f,\epsilon_2)  \le (1-\epsilon_1/\epsilon_2) \log |\mathcal{X} \times \mathcal{Y}|$.
\end{remark}

Note that  $\IC_\mu(f,\mu,0)$ is not always continuous with respect to $\mu$. For example, let
\[   \label{eq:measure-delta}
\mu_\epsilon =
\begin{pmatrix}
\frac{1-\epsilon}{3} & \frac{1-\epsilon}{3} \\
\frac{1-\epsilon}{3} & \epsilon
\end{pmatrix},  \qquad
\mu = \lim_{\epsilon \to 0} \mu_\epsilon = \begin{pmatrix}
\frac{1}{3} & \frac{1}{3} \\
\frac{1}{3} & 0
\end{pmatrix}.
\]
Now for the $2$-bit $\AND$ function, we have $\IC_{\mu}(\AND,\mu,0)=0$, while $\IC_{\mu_\epsilon}(\AND,\mu_\epsilon,0)=\IC_{\mu_\epsilon}(\AND)$ as $\mu_\epsilon$ has full support. Thus
$$\lim_{\epsilon \to 0} \IC_{\mu_\epsilon}(\AND, \mu_\epsilon,0) =\lim_{\epsilon \to 0} \IC_{\mu_\epsilon}(\AND) = \IC_{\mu}(\AND),$$
which is known to be bounded away from $0$. The same example also shows that $\IC_\mu(f, \mu, \epsilon)$ is not always continuous with respect to $\epsilon$ at $\epsilon = 0$ if $\mu$ depends on $\epsilon$. In fact, $\IC_{\mu_\epsilon}(\AND, \mu_\epsilon, \epsilon) = 0$ while $\IC_{\mu_\epsilon}(\AND, \mu_\epsilon, 0) = \IC_{\mu_\epsilon}(\AND)$ for all $\epsilon > 0$, hence when $\epsilon > 0$ is sufficiently small we find $\IC_{\mu_\epsilon}(\AND, \mu_\epsilon, 0)$ is close to $\IC_\mu(\AND)$ which is bounded away from $0$.

However it turns out that $\IC_\mu(f,\nu,\epsilon)$ is continuous with respect to $\mu$ for all $\epsilon \ge 0$ when $\nu$ is fixed.  This  follows from the fact, established in~\cite[Lemma 4.4]{SelfRed}, that for every protocol $\pi$ and every two measures $\mu_1$ and $\mu_2$ with $|\mu_1-\mu_2| \le \delta$ (the distribution metric is statistical distance), we have
\begin{equation}
\label{eq:continMu}
|\IC_{\mu_1}(\pi)-\IC_{\mu_2}(\pi)| \le 2 \log(|\cX \times \cY|) \delta + 2 H(2\delta).
\end{equation}
Consequently
$$|\IC_{\mu_1}(f,\nu,\epsilon)-\IC_{\mu_2}(f,\nu,\epsilon)| \le 2 \log(|\cX \times \cY|) \delta + 2 H(2\delta),$$
as $\IC_{\mu_1}(f,\nu,\epsilon) = \inf_{\pi} \IC_{\mu_1}(\pi)$ and $\IC_{\mu_2}(f,\nu,\epsilon)=\inf_\pi \IC_{\mu_2}(\pi)$ where both infimums are over all protocols $\pi$ that computes $[f, \nu, \epsilon]$. In particular, by taking $\epsilon = 0$ and $\nu$ to be a measure with full support, we get the following theorem.

\begin{theorem}[{\cite[Lemma 4.4]{SelfRed}}]  \label{thm:uniform-continuity-without-error}
$\IC_\mu(f)$ is uniformly continuous with respect to $\mu$.
\end{theorem}

Finally, note that if $|\nu_1-\nu_2| \le \delta \le \epsilon$, then
\begin{equation}
\label{eq:continNu}
\IC_{\mu}(f,\nu_1,\epsilon+\delta) \le \IC_{\mu}(f,\nu_2,\epsilon) \le  \IC_{\mu}(f,\nu_1,\epsilon-\delta).
\end{equation}
This proves the continuity with respect to $\nu$ when $\epsilon>0$. The following theorem summarizes the continuity of $\IC_{\mu}(f,\nu,\epsilon)$ with respect to its parameters.

\begin{theorem}[Uniform continuity with error]  \label{thm:uniform-continuity-with-error}
Consider $\delta>0$. For each $f\colon\mathcal{X} \times \mathcal{Y} \to \mathcal{Z}$, real numbers $\epsilon_2 \ge \epsilon_1 \ge \delta$, and measures $\mu_1,\mu_2,\nu_1,\nu_2$ on $\mathcal{X} \times \mathcal{Y}$ with $|\mu_1-\mu_2| \le \delta$ and $|\nu_1-\nu_2| \le \delta$, we have
$$|\IC_{\mu_1}(f,\nu_1,\epsilon_1)- \IC_{\mu_2}(f,\nu_2,\epsilon_2)| \le   \left(1-\frac{\epsilon_1}{\epsilon_2}+\frac{4\delta}{\epsilon_1}\right) \log |\mathcal{X} \times \mathcal{Y}|+2 H(2\delta).$$
\end{theorem}

\begin{proof}
By (\ref{eq:continuityEps}), we have
$$|\IC_{\mu_2}(f,\nu_2,\epsilon_2)- \IC_{\mu_2}(f,\nu_2,\epsilon_1)| \le (1-\frac{\epsilon_1}{\epsilon_2}) \log(|\cX \times \cY|).$$
By (\ref{eq:continNu}) and (\ref{eq:continuityEps}), we have
$$|\IC_{\mu_2}(f,\nu_2,\epsilon_1)- \IC_{\mu_2}(f,\nu_1,\epsilon_1)| \le(1-\frac{\epsilon_1-\delta}{\epsilon_1+\delta}) \log(|\cX \times \cY|) \le \frac{2\delta}{\epsilon_1} \log(|\cX \times \cY|) .$$
By (\ref{eq:continMu}), we have
$$|\IC_{\mu_2}(f,\nu_1,\epsilon_1)- \IC_{\mu_1}(f,\nu_1,\epsilon_1)| \le 2 \log(|\cX \times \cY|) \delta + 2 H(2\delta).$$
These three inequalities  imply the theorem.
\end{proof}

\subsection{The multiparty number-in-hand model}  \label{sec:multipartyModel}

The number-in-hand model is the most straightforward  generalization of Yao's two-party model to the settings where  more than two players are present. In this model there are $k$ players who wish to collaboratively compute a function $f\colon\cX_1 \times \ldots \times \cX_k \to \cZ$. The communication is in the shared blackboard model, which means that all the communicated bits are visible to all the players. Let $\mu$ be a probability distribution on $\cX_1 \times \ldots \times \cX_k$, and let $X=(X_1,\ldots, X_k)$ be sampled from $\cX_1 \times \ldots \times \cX_k$ according to $\mu$. Definition~\ref{def:infocost} generalizes in a straightforward manner to
$$\IC_\mu^\ext(\pi) = I(\Pi; X),$$
and
$$\IC_\mu(\pi) = \sum_{i=1}^k I(\Pi; X_{-i}|X_i),$$
where  $X_{-i}:=(X_1,\ldots ,X_{i-1},X_{i+1},\ldots, X_k)$. Note also that $I(\Pi; X|X_i) = I(\Pi;X_{-i}|X_i)$, and thus we have
$$\IC_\mu(\pi) = \sum_{i=1}^k I(\Pi; X|X_i).$$
The notations $\IC_\mu(f)$, $\IC_\mu(f,\epsilon)$, and $\IC_\mu(f,\nu,\epsilon)$, and the continuity results in Section~\ref{sec:IC-continuity} also generalize in a straightforward manner to this setting.

\section{The local characterization of the optimal information cost\label{sec:LocalConcav}}

We start by some definitions. Let $B$ be a random bit sent by one of the players, and let $\mu_0 = \mu|_{B=0}$ and $\mu_1 = \mu|_{B=1}$, or in other words
$$\mu_b(xy):= \Pr[XY=xy|B=b],$$
for $b=0,1$. Denote $\Pr[\cdot |xy] := \Pr[\cdot | XY=xy]$.

\begin{definition}
Let $\mu$ be a distribution and $B$ be a signal sent by one of the players. 
\begin{itemize}
\item $B$ is called \emph{unbiased} with respect to $\mu$ if $\Pr[B=0]=\Pr[B=1]=\frac{1}{2}$. 
\item $B$ is called \emph{non-crossing} if $\mu(xy) < \mu(x'y')$ implies that $\mu_b(xy) \le \mu_b(x'y')$ for $b=0,1$.
\item $B$ is called \emph{$\epsilon$-weak} if $|\Pr[B=0|xy]- \Pr[B=1|xy] | \le \epsilon$ for every input $xy$.
\end{itemize}
A protocol is said to be in \emph{normal form} with respect to $\mu$ if all its signals are unbiased and non-crossing with respect to $\mu$.
\end{definition}

Let $\Delta(\cX \times \cY)$ denote the set of distributions on $\cX \times \cY$. A measure $\mu \in \Delta(\cX \times \cY)$ is said to be \emph{internal-trivial} (\emph{resp. external-trivial}) for $f$ if $\IC_\mu(f)=0$ (resp. $\IC_\mu^\ext(f)=0$). These measures are characterized in \cite{DFHL}.

\subsection{The Local Characterization}   \label{sec:local-char}

Suppose that after a random bit $B$ is sent, if $B=0$, the players continue by running a protocol that is (almost) optimal for $\mu_0$, and if $B=1$, they  run a protocol that is (almost) optimal for $\mu_1$. Note that the amount of information that $B$ reveals about the inputs to an external observer is  $\MI(B;XY)$. This shows
\begin{equation}
\label{eq:localExternalUpper}
\IC_\mu^\ext(f) \le \MI(B;XY) + \Ex_B [\IC^\ext_{\mu_B}(f)],
\end{equation}
and similarly
\begin{equation}
\label{eq:localInternalUpper}
\IC_\mu(f) \le \MI(B;X|Y)+\MI(B;Y|X) +\Ex_B [\IC_{\mu_B}(f)].
\end{equation}
In~\cite{MR3210776} it is shown that these inequalities essentially characterize $\IC_\mu^\ext(f,\mu,0)$ and $\IC_\mu(f,\mu,0)$. Denote $\MI^\ext_B:=\MI(B;XY)$, and $\MI_B:=\MI(B;X|Y)+\MI(B;Y|X)$.

\begin{theorem}[\cite{MR3210776}]
\label{thm:LocalCharDistError}
Suppose that $C\colon\Delta(\cX \times \cY) \to [0,\log(| \cX \times \cY|)]$ satisfies
\begin{enumerate}[(i)]
\item $C(\mu)=0$ for every  measure $\mu$ such that $\IC_\mu(f,\mu,0) = 0$, and
\item for every signal $B$ that can be sent by one of the players
$$
\label{eq:localConc}
C(\mu) \le \MI_B +\Ex_B [ C(\mu_B)].
$$
\end{enumerate}
Then $C(\mu) \le \IC_\mu(f,\mu,0)$. Similarly if $\IC_\mu(f,\mu,0)$ is replaced by $\IC^\ext_\mu(f,\mu,0)$, and $\MI_B$ is replaced by $\MI^\ext_B$, then $C(\mu) \le \IC_\mu^\ext(f,\mu,0)$. Furthermore, in both of the external and the internal cases, it suffices to verify (ii) only for non-crossing unbiased signals $B$.
\end{theorem}

In light of Theorem~\ref{thm:LocalCharDistError}, in order to determine the values of $\IC_\mu(f,\mu,0)$, one has to first prove an upper bound by constructing a protocol (or a sequence of protocols) for every measure $\mu$. Then it suffices to verify that the bound satisfies the conditions of Theorem~\ref{thm:LocalCharDistError}.


In~\cite{MR3210776}, $\IC_\mu(\AND_2)$ is determined using Theorem~\ref{thm:LocalCharDistError}. However, the proof presented in~\cite{MR3210776}  contains some gaps. 
One error is the claim that it suffices to verify (\ref{eq:localConc}) for sufficiently weak signals. While it is not difficult to see that indeed it suffices to verify (\ref{eq:localConc}) for signals $B$ which are $\epsilon$-weak for an absolute constant $\epsilon>0$, in ~\cite{MR3210776} the condition (ii) is only verified for $\epsilon$ that is smaller than a function of $\mu$. This is not sufficient, and one can easily construct a counter-example  by allowing the signal $B$ to become increasingly weaker as $\mu$ moves closer to the boundary (by boundary we mean the set of measures $\mu$ that satisfy Theorem~\ref{thm:LocalCharDistError}~(i)). Indeed, for example, set $C(\mu)=K$ for a very large constant $K$ if $\mu$ does not satisfy Theorem~\ref{thm:LocalCharDistError}~(i), and otherwise set $C(\mu)=0$. Obviously (\ref{eq:localConc}) holds if $\mu$ is on the boundary. On the other hand, if $\mu$ is not on the boundary, then by taking $B$ to be sufficiently weak as a function of $\mu$, we can guarantee that $\mu_0$ and $\mu_1$ are not on the boundary either, and thus (\ref{eq:localConc}) holds in this case as well. However, taking $K$ to be sufficiently large violates the desired conclusion that $C(\mu) \le \IC_\mu(f,\mu,0)$.

To fix these errors, we start by observing that a straightforward adaptation of the proof of Theorem~\ref{thm:LocalCharDistError} yields an identical characterization of $\IC_\mu(f)$, however, with a different boundary condition.

\begin{theorem}
\label{thm:LocalCharNoError}
Suppose that $C\colon\Delta(\cX \times \cY) \to [0,\log(| \cX \times \cY|)]$ satisfies
\begin{enumerate}[(i)]
\item $C(\mu)=0$ for every measure $\mu$ such that $\IC_\mu(f) = 0$, and
\item for every signal $B$ that can be sent by one of the players
$$C(\mu) \le I_B +\Ex_B [ C(\mu_B)]. $$
\end{enumerate}
Then $C(\mu) \le \IC_\mu(f)$. Similarly, if we replace $\IC_\mu(f)$ by $\IC^\ext_\mu(f)$, and $I_B$ by $I^\ext_B$, then $C(\mu) \le \IC_\mu^\ext(f)$.
Furthermore, in both cases, it suffices to verify (ii) only for non-crossing unbiased signals $B$.
\end{theorem}

\begin{proof}
We only prove the internal case as the external case is similar. Let $\pi$ be a $c$-bit protocol in normal form and $\Pi$ be its transcript. For every possible transcript $t$, let $\mu_t:=\mu|_{\Pi=t}$.  Condition (ii) says that for every $1$-bit protocol $\pi$,
\begin{equation}
\label{eq:Lemma5.6II}
C(\mu) \le \IC_\mu(\pi) + \Ex_{t \sim \Pi}[C(\mu_t)].
\end{equation}
By a simple induction (see \cite[Lemma 5.6]{MR3210776}) this implies that (\ref{eq:Lemma5.6II}) holds for every $c$-bit protocol $\pi$ in a normal form where $c<\infty$.

Now consider an arbitrary protocol $\tau$ that computes $[f, 0]$. Lemma~\ref{lem:signal-simulation} below shows that one can simulate $\tau$ with a protocol $\pi$ that is in normal form and terminates with probability $1$. Note that $\pi$ also computes $[f, 0]$ and  by Proposition~\ref{prop:diff-protocol-same-cost} we have $\IC_\mu(\tau)=\IC_\mu(\pi)$.

Consider a large integer $c$, and let $\pi_c$ be the protocol that is obtained by truncating $\pi$ after $c$ bits of communication, clearly $\IC_\mu(\pi_c) \le \IC_\mu(\pi)$ as $\pi_c$ is a truncation of $\pi$. Let $G_c$ denote the set of leaves of $\pi_c$ in which the protocol is forced to terminate, and had we run $\pi$ instead, the communication would have continued. Let $\Pi_c$ denote the transcript of $\pi_c$. For any given $\delta>0$, one can guarantee that for every $xy$,
$$\Pr[\Pi_c(xy) \in G_c]<\delta$$
by choosing $c$ to be sufficiently large. As $\pi$ computes $[f, 0]$, for every leaf $t$ in $\pi_c$ such that $t \not\in G_c$, $\mu_t$ is an internal-trivial distribution, hence Condition (i) is satisfied on $\mu_t$ implying $C(\mu_t)=0$. Therefore (\ref{eq:Lemma5.6II}) shows
$$
C(\mu) \le \IC_\mu(\pi_c) + \delta \log(|\cX \times \cY|) \le \IC_\mu(\pi) + \delta \log(|\cX \times \cY|) = \IC_\mu(\tau) + \delta \log(|\cX \times \cY|).
$$
Letting $\delta \to 0$ one obtains the desired bound.
\end{proof}

We use the uniform continuity of $\IC_\mu(f)$ with respect to $\mu$ to prove that it suffices to verify Theorem~\ref{thm:LocalCharNoError}~(ii) for signals $B$ that are weaker than quantities that can depend on $\mu$. This as we shall see suffices to fix the proof of~\cite{MR3210776}.

\begin{theorem}[Main Theorem 1]
\label{thm:LocalCharNoErrorWeak}
Let $w\colon(0,1] \to (0,1]$ be a non-decreasing function, $\Omega \subseteq \Delta(\cX \times \cY)$ be a subset of measures containing the internal trivial distributions for function $f$. Let $\delta(\mu)$ denote the distance of $\mu$ from $\Omega$. Suppose that $C\colon\Delta(\cX \times \cY) \to [0,\log(| \cX \times \cY|)]$ satisfies
\begin{enumerate}[(i)]
\item $C(\mu)$ is uniformly continuous with respect to $\mu$;
\item $C(\mu)=\IC_\mu(f)$ if $\delta(\mu)=0$, and
\item for every non-crossing unbiased $w(\delta(\mu))$-weak signal $B$ that can be sent by one of the players,
\begin{equation}
\label{eq:localConcIII}
C(\mu) \le I_B +\Ex_B [ C(\mu_B)].
\end{equation}
\end{enumerate}
Then $C(\mu) \le \IC_\mu(f)$. Similarly, if we replace $\Omega$ by a subset containing the external trivial distributions for function $f$, in Condition (ii) replace $\IC_\mu(f)$ by $\IC^\ext_\mu(f)$, and in Condition (iii) replace $I_B$ by $I^\ext_B$, then $C(\mu) \le \IC_\mu^\ext(f)$.
\end{theorem}

The proof of Theorem~\ref{thm:LocalCharNoErrorWeak} is presented in Section~\ref{sec:fixed-proof}.

\subsection{The local characterization in a different form}   \label{sec:local-char-different}

Information cost measures the amount of information that is revealed by communicated bits. The local concavity conditions in Section~\ref{sec:local-char} become more natural if they are represented in terms of the amount of information that is \emph{not} revealed. Define the \emph{concealed information} and \emph{external concealed information} of a protocol $\pi$ with respect to $\mu$, respectively, as
$$\CI_\mu(\pi)  =H(X|\Pi Y)+H(Y|\Pi X) = H(X|Y)+H(Y|X)-\IC_\mu(\pi),$$
and
$$\CI^\ext_\mu(\pi) = H(XY|\Pi) = H(XY)-\IC^\ext_\mu(\pi),$$
where $\Pi$ is the transcript of $\pi$.

\begin{remark}
In the setting of multi-party number-in-hand model, we have
$$\CI_\mu(\pi)  = \sum_{i=1}^k H(X|\Pi X_i) =\left(\sum_{i=1}^k H(X|X_i)\right)  -\IC_\mu(\pi),$$
and
$$\CI^\ext_\mu(\pi) = H(X|\Pi) = H(X)-\IC^\ext_\mu(\pi),$$
where $X=(X_1,\ldots,X_k)$ is the random vector of all inputs.
\end{remark}

By using concealed information rather than information cost, the local characterization turns into a condition about the local concavity of the function.

\begin{lemma}
Inequalities \eqref{eq:localExternalUpper} and \eqref{eq:localInternalUpper} are respectively equivalent to
$$ \CI_\mu^\ext(f) \ge \Ex_B [ \CI^\ext_{\mu_B}(f)], \qquad \mbox{and} \qquad   \CI_\mu(f) \ge \Ex_B[ \CI_{\mu_B}(f)].$$
\end{lemma}
\begin{proof} Substituting $I(B;XY)=H(XY)-H(XY|B)$ in $\IC_\mu^\ext(f) \le I(B;XY) + \Ex_B [\IC^\ext_{\mu_B}(f)]$ leads to
$$\CI_\mu^\ext(f) \ge H(XY|B) - \Ex_B [H(XY|B=b) - \CI^\ext_{\mu_B}(f)]$$
which simplifies to the desired $\CI_\mu^\ext(f) \ge  \Ex_B [\CI^\ext_{\mu_B}(f)]$.

Similarly substituting $I(B;X|Y)+I(B;Y|X)=H(X|Y)- H(X|YB) + H(Y|X) -H(Y|XB)$ in $\IC_\mu(f) \le I(B;X|Y)+I(B;Y|X) +\Ex_B [\IC_{\mu_B}(f)]$ leads to
$$\CI_\mu(f) \ge H(X|YB)+ H(Y|XB) - \Ex_B[H(X|YB=b)+H(Y|XB=b)-\IC_{\mu_B}(f)]$$
which simplifies to  $\CI_\mu(f) \ge  \Ex_B[ \CI_{\mu_B}(f)]$.
\end{proof}

\section{Communication protocols as random walks on $\Delta(\cX \times \cY)$}   \label{sec:randomwalk}

Consider a protocol $\pi$ and a prior distribution $\mu$ on the set of inputs $\cX \times \cY$. Suppose that in the first round Alice  sends a random signal $B$ to Bob. We can interpret this as  a random update of  the prior distribution $\mu$ to a new distribution $\mu_0 = \mu|_{B=0}$ or $\mu_1 = \mu|_{B=1}$ depending on the value of $B$. It is not difficult to see that $\mu_b(x,y) =p_b(x) \mu(x,y)$ for $b=0,1$, where $p_b(x)=\frac{\Pr[B=b|x]}{\Pr[B=b]}$. In other words, $\mu_b$ is obtained by multiplying the rows of $\mu$ by non-negative numbers. Similarly if Bob is sending a message, then  $\mu_b$ is obtained by multiplying the columns of $\mu$ by the numbers $p_b(y)=\frac{\Pr[B=b|y]}{\Pr[B=b]}$. That is  $\mu_b(x,y) = \mu(x,y)p_b(y)$. Therefore, we can think of a protocol as a random walk on $\Delta(\cX \times \cY)$ that starts at $\mu$, and every time that a player sends a message, it moves to a new distribution. Note further that this random walk is without drift as $\mu=\Ex_B [ \mu_B]$.

Let $\Pi$ denote the transcript of the protocol. When the protocol terminates, the random walk stops at $\mu_\Pi := \mu|_\Pi$. Since $\Pi$ itself is a random variable, $\mu_\Pi$ is a random variable that takes values in  $\Delta(\cX \times \cY)$.  Interestingly, both the internal and external information costs of the protocol depend only on the distribution of $\mu_\Pi$ (this is a distribution on the set $\Delta(\cX \times \cY)$, which itself is a set of distributions). To see this, note  $\MI(X;\Pi|Y) = \Ex_{\pi \sim \Pi, y \sim Y}\Dvg(X|_{\Pi=\pi, Y=y} \| X|_{Y=y})$ and $\MI(XY;\Pi) = \Ex_{\pi \sim \Pi} \Dvg(XY|_{\Pi=\pi} \| XY)$, and thus both of these quantities are determined by $\mu$ and $\mu_\Pi$. This immediately leads to the following observation:

\begin{proposition} \cite{MarkComputable}   \label{prop:diff-protocol-same-cost}
\label{prop:randomWalk}
Let $\pi$ and $\tau$ be two communication protocols with the same input set $\cX \times \cY$ endowed with a probability measure $\mu$. Let $\Pi$ and $T$  denote the transcripts of $\pi$ and $\tau$, respectively. If $\mu_\Pi$ has the same distribution as  $\mu_T$, then $\IC_\mu(\pi)=\IC_\mu(\tau)$ and $\IC^\ext_\mu(\pi)=\IC^\ext_\mu(\tau)$.
\end{proposition}

Proposition~\ref{prop:randomWalk}  shows that in the context of information complexity, it does not matter how different the steps of two protocol are, and as long as they both  yield the same distribution on $\Delta(\cX \times \cY)$, they have the same internal and external information cost. Consequently, one can directly work with this random walk (or the resulting distribution on $\Delta(\cX \times \cY)$) instead of working with the actual  protocols. Indeed, let $\mathcal{C}^T_\mu(\Delta(\cX \times \cY))$ denote the set of all probability distributions on  $\Delta(\cX \times \cY)$ that can be obtained, starting from the distribution $\mu$, through communication protocols that perform a given communication task $T$. The information cost of performing the task $T$ is the infimum of the information costs of the distributions in  $\mathcal{C}^T_\mu(\Delta(\cX \times \cY))$. Although, as mentioned earlier, this infimum is not always attained, if one takes the closure of $\mathcal{C}^T_\mu(\Delta(\cX \times \cY))$ (under weak convergence) then one can replace the infimum with
minimum. For the $2$-bit $\AND$ function, the buzzers protocol of \cite{MR3210776} yields the distribution in the closure of $\mathcal{C}^T_\mu(\Delta(\cX \times \cY))$ that achieves the minimum information cost. The buzzers protocol is not a communication protocol, but one can consider it as the limit of a sequence of communication protocols. We believe that the following is an important open problem.

\begin{problem}
Define a  paradigm such that for every communication task $T$ and every measure $\mu$ on an input set $\cX \times \cY$, the set of distributions on $\Delta(\cX \times \cY)$ resulting from the protocols performing the task $T$ in this paradigm is exactly equal to  the closure of $\mathcal{C}^T_\mu(\Delta(\cX \times \cY))$.
\end{problem}

Partial progress towards resolving this problem has been made in ~\cite{DaganFilmus}, see also~\cite{DFHL}.

\subsection{A signal simulation lemma}  \label{sec:signal-simulation-lemma}

Here we prove a simulation lemma that will be useful in the proof of the local characterization theorems. We start by restating a splitting lemma from \cite{MR3210776}. We use the notation $[\mu_0,\mu_1]$ for the set of all convex combinations $\alpha \mu_0 + (1-\alpha) \mu_1$, where $\alpha \in [0,1]$.

\begin{lemma}[Splitting Lemma,~\cite{MR3210776}]
\label{lem:Splitting}
Consider $\mu \in \Delta(\cX \times \cY)$ and a  signal $B$ sent by one of the players, and let $\mu_b=\mu|_{B=b}$ for $b=0,1$. Consider $\rho_0,\rho_1 \in [\mu_0,\mu_1]$ and $\rho \in (\rho_0,\rho_1)$. There exists a signal $B'$ that the same player can send starting at $\rho$ such that $\rho_b=\rho|_{B'=b}$ for $b=0,1$.
\end{lemma}

Lemma \ref{lem:Splitting} is proved in~\cite[Lemma 7.11]{MR3210776}, there is a minor error in the original statement as it is claimed that the lemma holds for $\rho \in [\rho_0,\rho_1]$ where the interval is closed.

We are now ready to prove the signal simulation lemma, which says every signal can be perfectly simulated by a non-crossing unbiased $\epsilon$-weak signal sequence. This lemma generalizes \cite[Lemma 5.2]{MR3210776}.

\begin{lemma}[Signal Simulation]   \label{lem:signal-simulation}
Let $\epsilon>0$, and consider $\mu \in \Delta(\cX \times \cY)$ and a  signal $B$ sent by  one of the players. There exists a sequence of non-crossing unbiased  $\epsilon$-weak random signals $\mathcal{B}=(B_1 B_2 \ldots)$ that with probability $1$ terminates, and furthermore $\mu|_{\mathcal{B}}$ has the same distribution as $\mu|_B$.
\end{lemma}

\begin{proof}
Let $\mu_0=\mu|_{B=0}$ and $\mu_1=\mu|_{B=1}$. The following protocol explains how the sequence $(B_1 B_2 \ldots)$ is constructed from the signal $B$.

\begin{framed}
\begin{itemize}
\item Set $\mu_c = \mu$ and $i=1$;
\item Repeat until $\mu_c=\mu_0$ or $\mu_c=\mu_1$;
\item \qquad If $\mu_c \in [\mu_0,\mu]$, then
\item \qquad \qquad Set $\lambda$ to be the largest value in $[0,1]$ satisfying
\item \qquad \qquad \qquad $\lambda \max_{xy} \frac{|\mu_c(x,y)-\mu_0(x,y)|}{\mu_c(x,y)} \le \epsilon$, and
\item \qquad \qquad \qquad $\lambda |\mu_0(x,y)-\mu_0(x',y')-\mu_c(x,y)+\mu_c(x',y')| \le \mu_c(x',y') - \mu_c(x,y)$ if $\mu_c(x,y)<\mu_c(x',y')$.
\item \qquad \qquad Send a signal $B_i$ that splits $\mu_c$ to $(1-\lambda) \mu_c+\lambda \mu_0$ and $(1+\lambda) \mu_c -\lambda \mu_0$;
\item \qquad If $\mu_c \in (\mu,\mu_1]$, then
\item \qquad \qquad Set $\lambda$ to be the largest value in $[0,1]$ satisfying
\item \qquad \qquad \qquad $\lambda \max_{xy} \frac{|\mu_c(x,y)-\mu_1(x,y)|}{\mu_c(x,y)} \le \epsilon$, and
\item \qquad \qquad \qquad $\lambda |\mu_1(x,y)-\mu_1(x',y')-\mu_c(x,y)+\mu_c(x',y')| \le \mu_c(x',y') - \mu_c(x,y)$ if $\mu_c(x,y)<\mu_c(x',y')$.
\item \qquad \qquad Send a signal $B_i$ that splits $\mu_c$ to $(1-\lambda) \mu_c+\lambda \mu_1$ and $(1+\lambda) \mu_c -\lambda \mu_1$;
\item \qquad Update $\mu_c$ to the current distribution;
\item \qquad Increase $i$;
\end{itemize}
\end{framed}

Note that every signal $B_i$ sent in the above protocol is $\epsilon$-weak and non-crossing. Indeed, if $B_i$ splits $\mu_c$ into $(1-\lambda) \mu_c+\lambda \mu_0$ and $(1+\lambda) \mu_c -\lambda \mu_0$, then
\begin{align*}
|\Pr[B_i=0 | xy]-\Pr[B_i=1 | xy]| 
&= \left|\frac{\mu_c(xy | B_i=0)}{2\mu_c(xy)} - \frac{\mu_c(xy | B_i=1)}{2\mu_c(xy)} \right| \\
&= \lambda \frac{|\mu_c(xy)-\mu_0(xy)|}{\mu_c(xy)},
\end{align*}
and the choice of $\lambda$ guarantees that this is bounded by $\epsilon$. The same calculation shows the $\epsilon$-weakness for $\mu_c \in [\mu,\mu_1]$. It can also be easily verified that the signal is  non-crossing.

To see that this sequence terminates with probability $1$, define
\begin{align*}
\Omega
&= \{ \nu \in [\mu_0,\mu_1] : \exists\ (x,y),(x',y') \ \text{s.t.} \  \nu(x,y)=\nu(x',y'), \phantom{\}} \\
&\phantom{\quad \{} \text{while}\  \mu_0(x,y) \neq \mu_0(x',y') \ \text{or}\  \mu_1(x,y) \neq \mu_1(x',y') \},
\end{align*}
and notice that $\Omega$ is a finite set. Consider $\mu_c \in [\mu_0,\mu]$. If the value of $\lambda$ is set by the first condition, then there is a uniform lower-bound for $\lambda$:
$$\lambda \ge \lambda_0:= \epsilon /\max_{xy} \frac{|\mu(x,y)-\mu_0(x,y)|}{\mu(x,y)} =\epsilon  / \max_{xy} \frac{|\mu(x,y)-\mu_1(x,y)|}{\mu(x,y)}>0.$$
Moreover if $\lambda$ is set by the other  condition, then it means $\mu_c(x,y)<\mu_c(x',y')$, and  at least one of $\mu_c|_{B_i=0}$ or $\mu_c|_{B_i=1}$  belongs to $\Omega$. Hence starting at any point $\mu_c$, the random walk terminates with probability at least $2^{-\lceil 1/\lambda_0 \rceil + |\Omega|}$ after $\lceil 1/\lambda_0 \rceil+ |\Omega|$ steps. It follows that with probability $1$, the random walk terminates.
\end{proof}

\subsection{Proofs of Theorem~\ref{thm:LocalCharNoErrorWeak}}  \label{sec:fixed-proof}

We present the proof for the internal case only as the external case is similar.

\begin{lemma}  \label{lem:induction-lemma}
Let $w, \delta(\mu)$ and $C$ be as in Theorem \ref{thm:LocalCharNoErrorWeak}, and suppose $C$ satisfies Conditions (i), (ii) and (iii). Let $\tau$ be a protocol that terminates with probability $1$, and further assume $\tau$ is in normal form and every signal sent in $\tau$ is $\epsilon$-weak. Given a probability distribution $\mu \in \Delta(\cX \times \cY)$, for every node $u$ in the protocol tree of $\tau$, let $\mu_u$ be the probability distribution conditioned on the event that the protocol reaches $u$. If $\mu$ satisfies $w(\delta(\mu_u)) \ge \epsilon$ for every internal node $u$, then
$$C(\mu) \le \IC_\mu(\tau) + \Ex_{\ell} [C(\mu_\ell)],$$
where the expected value is over all leaves $\ell$ of $\tau$ chosen according to the distribution (on the leaves) when the inputs are sampled according to $\mu$.
\end{lemma}

\begin{proof}
For every internal node $u$, the assumption in the statement of the lemma implies that the signal sent from $u$ is $w(\delta(\mu_u))$-weak. Hence Condition (iii) shows that the claim is true if $\tau$ is a $1$-bit protocol, and thus by a simple induction (See \cite[Lemma 5.6]{MR3210776}) it is true if $\tau$ is a $c$-bit protocol for any $c < \infty$.

Now assume $\tau$ has infinite depth. Consider a large integer $c$, obtain $\tau_c$ by truncating $\tau$ after $c$ bits of communication, trivially $\IC_\mu(\tau_c) \le \IC_\mu(\tau)$. Let $G_c$ denote the set of the leaves of $\tau_c$ in which the protocol is forced to terminate. Let $\cL_c$ be the set of leaves in $\tau$ with depth at most $c$. Clearly, the set of leaves in $\tau_c$ is exactly $G_c \cup \cL_c$. As $\tau_c$ has a bounded depth, we have
$$
C(\mu) \le \IC_\mu(\tau_c) + \Ex_{\ell \in \cL_c \cup G_c} [C(\mu_\ell)] \le \IC_\mu(\tau) + \Ex_{\ell \in \cL_c \cup G_c} [C(\mu_\ell)].
$$
Let $\Pi_c$ denote the transcript of $\tau_c$. As $\tau$ terminates with probability $1$, given any $\alpha > 0$, one can guarantee $\Pr[\Pi_c(xy) \in G_c]<\alpha$ for every $xy$ by choosing $c$ to be sufficiently large. Hence
$$
C(\mu) \le \IC_\mu(\tau) + \Ex_{\ell \in \cL_c} [C(\mu_\ell)] + \alpha \log(|\cX \times \cY|).
$$
Taking the limit $\alpha \to 0$ shows $C(\mu) \le \IC_\mu(\tau) + \Ex_{\ell \in \cL} [C(\mu_\ell)]$ where $\cL$ is the set of all leaves of $\tau$.
\end{proof}

\begin{proof}[Proof of Theorem \ref{thm:LocalCharNoErrorWeak}]
Firstly by (ii), $\delta(\mu) = 0$ implies $C(\mu) = \IC_\mu(f) \le \IC_\mu(f)$. Hence assume $\delta(\mu) > 0$. Consider an arbitrary signal $B$ sent by Alice.  As we discussed before, one can interpret $B$ as a one step random walk that starts  at $\mu$ and jumps  either to $\mu_0$ or to $\mu_1$ with corresponding probabilities $\Pr[B=0 | X=x]$ and $\Pr[B=1 | X=x]$. The idea behind the proof is to use Lemma \ref{lem:signal-simulation} to simulate this random jump with a random walk that has smaller steps so that we can apply the concavity assumption of the theorem to those steps.

Let $\pi$ be a protocol that  computes $[f, 0]$. For $0< \eta < \delta(\mu)$, applying Lemma~\ref{lem:signal-simulation} one gets a new protocol $\tilde{\pi}$ by replacing every signal sent in $\pi$ with a random walk consisting of $w(\eta)$-weak non-crossing unbiased signals. Note $\tilde{\pi}$ terminates with probability $1$. Moreover, since $\tilde{\pi}$ is a perfect simulation of $\pi$, by Proposition~\ref{prop:diff-protocol-same-cost} we have $\IC_\mu(\pi)=\IC_\mu(\tilde{\pi})$.

For every node $v$ in the protocol tree of $\tilde{\pi}$, let $\mu_v$ be the measure $\mu$ conditioned on the event that the protocol reaches the node $v$. Obtain $\tau$ from $\tilde{\pi}$ by terminating at every node $v$ that satisfies $\delta(\mu_v) \le \eta$. Note that by the construction, Condition (iii) is satisfied on every internal node $v$ of $\tau$, as every such node satisfies $\eta < \delta(\mu_v)$, thus $w(\eta) \le w(\delta(\mu_v))$ implying the signal sent on node $v$ is $w(\delta(\mu_v))$-weak.  Hence by Lemma~\ref{lem:induction-lemma}, 
$$C(\mu) \le \IC_\mu(\tau) + \Ex_{\ell} [C(\mu_\ell)],$$
where the expected value is over all leaves of $\tau$. For every $\mu_\ell$, let $\mu'_\ell \in \Omega$ be a distribution such that $\delta(\mu_\ell) = |\mu_\ell - \mu'_\ell|$. By Conditions (i) and (ii), and the uniform continuity of $\IC_\mu(f)$, we have that for every $\epsilon > 0$ there exists $\eta > 0$, such that for all $\mu_\ell$, as long as $\delta(\mu_\ell) = |\mu_\ell - \mu'_\ell| \le \eta$, then
\[   \nonumber
C(\mu_\ell) \le C(\mu'_\ell) + \epsilon = \IC_{\mu'_\ell}(f) + \epsilon \le \IC_{\mu_\ell}(f) + \epsilon + \epsilon = \IC_{\mu_\ell}(f) + 2\epsilon.
\]
As a result,
\[ \nonumber
C(\mu) \le \IC_\mu(\tau) + \Ex_{\ell} [\IC_{\mu_\ell}(f) + 2\epsilon] = \IC_\mu(\tau) + \Ex_{\ell} [\IC_{\mu_\ell}(f)] + 2\epsilon.
\]
Since $\mu_\ell$ is generated by truncating $\tilde{\pi}$, we have
\[ \nonumber
\IC_\mu(\tau) + \Ex_{\ell} [\IC_{\mu_\ell}(f)] \le \IC_\mu(\tilde{\pi}) = \IC_\mu(\pi).
\]
Therefore $C(\mu) \le \IC_\mu(\pi) + 2\epsilon$. As this holds for arbitrary $\epsilon$, we must have $C(\mu) \le \IC_\mu(\pi)$.
\end{proof}

\section{The multiparty $\AND$ function in the number-in-hand model}

In~\cite{MR3210776} it is shown that in the two-player setting, a certain (unconventional) protocol that we refer to as the buzzers protocol, has optimal information and external information cost for the $2$-bit $\AND$ function. In this section we show that the buzzers protocol can be generalized to an optimal protocol for the multi-party $\AND$ function in the number-in-hand model (assuming Assumption~\ref{assumption}).

For the sake of brevity, we denote $\mu_x:=\mu(\{x\})$  for every $x \in \{0,1\}^k$. Furthermore we assume that $\mu_{\e_1} \le \ldots \le \mu_{\e_k}$. The protocol is given by having buzzers with waiting times which have independent exponential distributions, and start at times $t_1,\ldots,t_k$  for players $1,\ldots,k$, respectively.  Although the protocol $\pi_\mu^\wedge$ described in Figure~\ref{fig:prot} is not a conventional communication protocol, it can be easily approximated by discretization and truncation of  time.

\begin{figure}[h]
\caption{The  protocol $\pi_\mu^\wedge$ for solving the $\AND$ function on a distribution $\mu$.\label{fig:prot}}
\begin{framed}
\begin{itemize}
\item There is a  clock whose time starts at $0$ and  increases continuously to $+\infty$.
\item  Let $t_i:= \ln \frac{\mu_{\e_i}}{\mu_{\e_1}}$ for $i=1,\ldots,k$, and let $t_{k+1}:=\infty$.

\item For every $i=1,\ldots,k$, if $x_i=0$, then the $i$-th player privately picks an independent random variable $T_i$ with exponential distribution, and if time reaches $t_i+T_i$, the player announces that his/her input is $0$, and the protocol terminates immediately with all the players knowing that $\bigwedge_{i=1}^k x_i=0$.
\item If the clock reaches $+\infty$ without any player announcing their input, the players will know that $\bigwedge_{i=1}^k x_i=1$.
\end{itemize}
\end{framed}
\end{figure}

Recall that the exponential distribution is memoryless, and intuitively can be generated in the following manner: Consider a buzzer starting at time $t=0$. At every infinitesimal time interval of length $dt$, independently of the past, it buzzes with probability $dt$, and then stops. The waiting time for a buzz to happen has exponential distribution.

Thus we can describe $\pi_\mu^\wedge$ as in the following: For every $i \in [k]$,  if $x_i=0$, the $i$-th player activates a buzzer at time $t_i$. When the first buzz happens the protocol terminates, and the players decide that $\bigwedge_{i=1}^k x_i = 0$. If the time reaches $\infty$ without anyone buzzing, they decide $\bigwedge_{i=1}^k x_i = 1$. In Theorem~\ref{thm:mainMultiparty} we show that for the measures $\mu$ that satisfy Assumption~\ref{assumption}, the protocol $\pi^\wedge_\mu$ has optimal external and internal information cost.

\begin{theorem}[Main Theorem 2]
\label{thm:mainMultiparty}
For every $\mu$ satisfying Assumption~\ref{assumption}, the protocol $\pi^\wedge_\mu$ has the smallest external and internal information cost.
\end{theorem}

In order to prove Theorem~\ref{thm:mainMultiparty}, we need to verify that the concavity conditions of Theorem~\ref{thm:LocalCharNoErrorWeak} are satisfied. Consider the measure $\mu$ that is uniformly distributed over $\e_1,\ldots,\e_k$. That is $\mu_{\vec{0}} = 0$ and $\mu_{\e_1} = \ldots = \mu_{\e_k} = 1/k$. Note that when the protocol $\pi^\wedge$ is executed on this measure, all the players become active at time $0$, and as time proceeds, they do not obtain any new information about the inputs of the other players until one of the players buzzes. Then at that point the input is revealed to all the players. Due to this discrete nature of the corresponding random-walk on $\Delta(\cX \times \cY)$, we need to analyze this particular measure separately, and afterwards when verifying the concavity conditions, we can let $\Omega$ in the statement of Theorem~\ref{thm:LocalCharNoErrorWeak} include this measure. Claim~\ref{claim:protocol-is-opt-when-mu-i-are-all-equal} below verifies Theorem~\ref{thm:mainMultiparty} for this particular measure.

\begin{myclaim}   \label{claim:protocol-is-opt-when-mu-i-are-all-equal}
Let $\mu$ be the measure that $\mu_{\vec{0}} = 0$ and $\mu_{e_1} = \ldots = \mu_{e_k} = 1/k$.
The internal and external information cost of the protocol $\pi^\wedge$ is optimal with respect to $\mu$.
\end{myclaim}

\begin{proof}
First we  present the proof for the  external information complexity. Let $\pi$ be any protocol that solves the multi-party AND function correctly, and let $t$ be a possible transcript of the protocol. First note that it is not possible to have $\Pr[\Pi_{e_m} = t] > 0$  for all $1 \le m \le k$. Indeed by rectangle property this would imply $\Pr[\Pi_{\vec{1}} = t] > 0$, and since  the correct output for $\vec{1}$ is different from that of $e_1,\ldots,e_k$, we would get a contradiction with the assumption that $\pi$ solves AND correctly on all inputs. Hence to every transcript $t$, we can assign a $j(t) \in \{1,\ldots,m\}$ with $\Pr[\Pi_{e_j}=t]=0$. Now for a random $X \sim \mu$, denote $J=j(\Pi_X)$, and notice that conditioned on $J=j$, $X$ is supported  on the set $\{e_1,\ldots,e_k\} \setminus \{e_j\}$ of size $k-1$, and thus $\entropy(X|_J) \le \log (k-1)$. Consequently, we have 
$$\IC_{\mu}^\ext(\pi)=\MI(X;\Pi_X) = \MI(X;\Pi_X J) \ge \MI(X; J) = \entropy(X)-\entropy(X|J) \ge \log k - \log (k-1).$$

On the other hand, consider our protocol $\pi_\mu^\wedge$. Note that under $\mu$, all players are activated at the same time, and consequently by symmetry, for every termination time $\tau$ and player $j \in \{1,\ldots,k\}$, the random variable $X|_{\Pi_X=(\tau,j)}$ is uniformly distributed on $\{e_1,\ldots,e_k\} \setminus \{e_j\}$. Hence $\entropy(X| \Pi_X)= \log(k-1)$. We conclude that 
$$\IC_{\mu}^\ext(\pi^\wedge)=\MI(X;\Pi_X) = \entropy(X) - \entropy(X|\Pi_X) = \log k - \log (k-1).$$

Next we turn to the internal case. Again, let $\pi$ be any protocol that solves the multi-party AND correctly, and let $J$ be defined as above. First note that for $i \in [k]$, $X|_{X_i=1}$ is supported on the single point $\{e_i\}$ and  $X|_{X_i=0}$ is uniformly distributed on $\{e_1,\ldots,e_k\} \setminus \{e_i\}$. Hence
$$\entropy(X|X_i)= \frac{1}{k} \entropy(X|X_i=1)+ \frac{k-1}{k} \entropy(X|X_i=0)=\frac{k-1}{k} \log(k-1).$$
Moreover for $i,j \in [k]$, $X|_{J=j,X_i=0}$ is supported on $\{e_1,\ldots,e_k\} \setminus \{e_i,e_j\}$. Hence using $\Pr[J=i]=\Pr[J=i, X_i=0]$, we have 
\begin{align*}
\entropy(X|JX_i) &=  \sum_{j=1}^k \Pr[J=j, X_i=0]  \entropy(X|_{J=j,X_i=0}) \\ 
&\le \sum_{j=1}^k  \Pr[J=j, X_i=0]  \log |\{e_1,\ldots,e_k\} \setminus \{e_i,e_j\}| \\
&= \frac{k-1}{k} \log(k-2) + \Pr[J=i] (\log(k-1)-\log(k-2)).
\end{align*}
Summing over $i$, we obtain
$$ \sum_{i=1}^k \entropy(X|JX_i) = (k-2) \log(k-2) + \log(k-1), $$
and thus
\begin{align*}
\IC_{\mu}(\pi)
&=\sum_{i=1}^k \MI(X;\Pi_X|X_i) = \sum_{i=1}^k \MI(X;\Pi_X J|X_i)\\ 
&\ge \sum_{i=1}^k \MI(X; J|X_i) = \sum_{i=1}^k \entropy(X|X_i)-\entropy(X|JX_i) \\ 
&\ge (k-1) \log(k-1) -  ((k-2) \log(k-2)+\log(k-1)) \\
&= (k-2)(\log(k-1)-\log(k-2)).
\end{align*}

On the other hand, for the protocol $\pi_\mu^\wedge$, by symmetry, for every termination time $\tau$ and player $j \in \{1,\ldots,k\}$, the random variable $X|_{\Pi_X=(\tau,j), X_i=0}$ is uniformly distributed on $\{e_1,\ldots,e_k\} \setminus \{e_j,e_i\}$. Hence 
$$ \entropy(X| \Pi_X X_i)=\frac{1}{k} \log(k-1) + \frac{k-2}{k} \log(k-2).$$
We conclude that 
\begin{align*}
\IC_{\mu}(\pi^\wedge)
&= \sum_{i=1}^k \MI(X;\Pi_X|X_i) = (k-1)\log(k-1) - \sum_{i=1}^k \entropy(X|\Pi_X X_i)\\
&= (k-2)(\log(k-1)-\log(k-2)).  \qedhere 
\end{align*}
\end{proof}

\begin{myclaim}
\label{claim:No1}
It suffices to verify Theorem~\ref{thm:mainMultiparty} for measures $\mu$ with $\mu(\vec{1})=0$. 
\end{myclaim} 

\begin{proof}
Let $\mu$ be a measure satisfying Assumption~\ref{assumption}, and let $\pi$ be a protocol that solves the multiparty AND function correctly on all the inputs.  Let $\Pi$ denote the transcript of this protocol, and let $B=1_{[X = \vec{1}]}$. Since $\pi$ solves the AND function correctly,  $\Pi$ determines the value of $B$. We have 
\begin{align*}
\IC_\mu^\ext(\pi)&=\MI(X;\Pi_X) = \MI(X B_X;\Pi_X) = \MI(B_X;\Pi_X)+ \MI(X;\Pi_X|B_X) \\
&=0+ \Pr[X=\vec{1}]\MI(X;\Pi_X|X=\vec{1})+ \Pr[X \neq \vec{1}]\MI(X;\Pi_X|X \neq \vec{1}) \\
&=\Pr[X \neq \vec{1}]\MI(X;\Pi_X|X \neq \vec{1})=(1-\mu_{\vec{1}}) \IC_{\mu'}^\ext(\pi),
\end{align*}
where $\mu'$ is the measure $\mu$ conditioned on the event that the input is not  equal to $\vec{1}$. Similarly 
$$\IC_\mu(\pi)=(1-\mu_{\vec{1}})\IC_{\mu'}(\pi).$$
Finally, to conclude the claim, note that $\pi_\mu^\wedge$ and $\pi_{\mu'}^\wedge$ are identical as $\mu_{\e_i}/\mu_{\e_1}=\mu'_{\e_i}/\mu'_{\e_1}$ for all $i=1,\ldots,k$.  
\end{proof}

\section{Proof of Theorem~\ref{thm:mainMultiparty}}

In this section we prove Theorem~\ref{thm:mainMultiparty} by verifying the concavity conditions of Theorem~\ref{thm:LocalCharNoErrorWeak}. Let $\mu$ be a measure satisfying Assumption~\ref{assumption},  and $X=(X_1,\dots, X_k)$ denote the  random $k$-bit input. Let $\Pi$ be the random variable corresponding to the transcript of the protocol $\pi^\wedge_{\mu}$. Let $\Pi_x=\Pi|_{X=x}$.

To verify the concavity condition, we consider a signal $B$  with parameter $\epsilon$ sent by the player $s$. That is
$$\Pr[B=0 | X_s=0]=\frac{1+\epsilon \Pr[X_s=1]}{2},$$ 
and
$$\qquad \Pr[B=1 | X_s=1]=\frac{1+\epsilon \Pr[X_s=0]}{2}.$$
Note that $\Pr[B=0]=\Pr[B=1]=\frac{1}{2}$, i.e., the signal $B$ is unbiased. Let $\mu^0$ and $\mu^1$ respectively denote the distributions of $X^0 := X|_{B=0}$ and $X^1 := X|_{B=1}$. We have $\mu= \frac{\mu^0 + \mu^1}{2}$. Let $\Pi^0$ and $\Pi^1$ denote the random variables corresponding to the transcripts of $\pi^\wedge_{\mu^0}$ and $\pi^\wedge_{\mu^1}$, respectively.

Note that the transcript of $\pi_\mu^\wedge$ contains the termination time $t$, and if $t<\infty$, also the name of the player who first buzzed. We denote by $\pi_\infty$ the transcript corresponding to termination time $t=\infty$, and by $\pi_t^m$ the termination time $t<\infty$ with the $m$-th player buzzing.
 
 For $t \in [0,\infty)$, let $\Phi_x(t)$ denote the total amount of active time spent by all players before time $t$ if the input is $x$. For  $t_r \le t < t_{r+1}$, we have
$$\Phi_x(t) =\sum_{i: x_i=0} \max(t-t_i,0) = \sum_{i \in [1,r], x_i=0} t-t_i.$$
The probability density function $f_x$ of $\Pi_x$ is given by
$$f_{x}(\pi_t^m) = \left\{
\begin{array}{lcl}
0 &\qquad& \mbox{$t_m>t$ or $x_m=1$}  \\
e^{-\Phi_x(t)} & & \mbox{otherwise}
\end{array}
\right.,
$$
and $\Pr(\Pi_{\vec{1}}=\pi_\infty)=1$. The distribution of the transcript $\Pi$ is then
$$f(\pi_t^m) =\sum_x \mu_x f_{x}(\pi_t^m).$$
Define $f^0$, $f^0_x$ and $f^1$, $f^1_x$ analogously for $\pi^{\wedge}_{\mu^0}$ and $\pi^{\wedge}_{\mu^1}$, respectively.

\subsection{Probability distributions $\mu^0$ and $\mu^1$} \label{sec:distribution-mu0-mu1}
Denote $\beta_s :=\Pr[X_s=1]$, and $\zeta_s :=\Pr[X_s=0]$. For $B=0$, we have,
\begin{align*}
\mu^{0}_x =& \left\{
\begin{array}{lcl}
(1+\epsilon \Pr[x_s=1]) \mu_x = (1+\epsilon \beta_s) \mu_x  &\qquad& x_s=0  \\
(1-\epsilon \Pr[x_s=0]) \mu_x = (1-\epsilon \zeta_s) \mu_x & & x_s=1
\end{array}
\right.
\end{align*}
Consequently, the new starting times are $t^0_i=t_i$ for $i \neq s$, and $t^0_s= t_s - \gamma_0$ where
\[    \label{eqn:gamma-0}
\gamma_0=\ln\left(\frac{1+\epsilon \beta_s}{1-\epsilon \zeta_s}\right).
\]
Hence
$$\mu^{0}_x f^0_x(\pi_t^m) =
\left\{
\begin{array}{lcl}
\mu^{0}_x f_x(\pi_t^m) &\qquad  & t<t_s-\gamma_0 \\
(1+\epsilon \beta_s)\mu_x f_x(\pi_t^m) e^{-(t-t_s+\gamma_0)} & & t \in [t_s-\gamma_0,t_s), x_s=0, m \neq s \\
(1-\epsilon \zeta_s) \mu_x f_x(\pi_t^m)  & & t \in [t_s-\gamma_0,t_s), x_s=1, m \neq s \\
\mu^{0}_x f^0_x(\pi_t^s) & & t \in [t_s-\gamma_0,t_s), m=s \\
(1-\epsilon \zeta_s) \mu_x f_x(\pi_t^m) & & t\ge t_s \\
\end{array}
\right.
$$

On the other hand, for $B=1$, we have,
\begin{align*}
\mu^{1}_x =& \left\{
\begin{array}{lcl}
(1-\epsilon \beta_s) \mu_x &\qquad& x_s=0  \\
(1+\epsilon \zeta_s) \mu_x & & x_s=1
\end{array}
\right.
\end{align*}
Consequently the new starting times are $t^1_i=t_i$ for $i \neq s$, and $t^1_s= t_s + \gamma_1$ where
\[    \label{eqn:gamma-1}
\gamma_1=\ln\left(\frac{1+\epsilon \zeta_s}{1-\epsilon \beta_s}\right).
\]
Hence when $m \neq s$,
$$\mu^{1}_x f^1_x(\pi_t^m) =
\left\{
\begin{array}{lcl}
\mu^{1}_x f_x(\pi_t^m) &\qquad  & t\le t_s \\
(1- \epsilon \beta_s)\mu_x f_x(\pi_t^m) e^{t-t_s} & & t \in [t_s,t_s+\gamma_1), x_s=0 \\
(1+\epsilon \zeta_s) \mu_x f_x(\pi_t^m)  & & t \in [t_s,t_s+\gamma_1), x_s=1 \\
(1+\epsilon \zeta_s) \mu_x f_x(\pi_t^m) & & t\ge t_s+\gamma_1 \\
\end{array}
\right.
$$
For $m = s$,
$$\mu^{1}_x f^1_x(\pi_t^s) =
\left\{
\begin{array}{lcl}
(1+\epsilon \zeta_s) \mu_x f_x(\pi_t^s)  &\qquad&  t > t_s + \gamma_1 \text{\ and\ } x_s=0  \\
0  & & \text{otherwise}.
\end{array}
\right.
$$

\subsection{Setting up and first reductions} \label{sec:Some-reductions}

We set $\Omega$ to be the set of all external (resp. internal) trivial measures together with the measure in Claim~\ref{claim:protocol-is-opt-when-mu-i-are-all-equal}, and in the external case we set $w(x) = ck^{-20} x^4$ for some fixed constant $c > 0$ (one may need to pick a different $w(x)$ for internal case). Using the memoryless property of exponential distribution, we can shift the activation time of all the players by $-\ln(\mu_{\e_s}/\mu_{\e_1})$, and assume that $t_1=-\ln(\mu_{\e_s}/\mu_{\e_1}), \ldots, t_s=0, \ldots, t_k=\ln(\mu_{\e_k}/\mu_{\e_s})$.
 
Let $\phi(x) := x \ln (x)$. Using the notion of concealed information from Section~\ref{sec:local-char-different}, the concavity conditions of Theorem~\ref{thm:LocalCharNoErrorWeak} reduce to verifying 
\begin{align}
\int_{-\infty}^\infty &\sum_m \left(   \phi(f(\pi^m_t)) - \frac{\phi(f^0(\pi^m_t))+\phi(f^1(\pi^m_t))}{2}   \right) \nonumber \\ 
&-\sum_m \sum_x  \left( \phi(\mu_x f_x(\pi^m_t))-\frac{\phi(\mu^0_x f^0_x(\pi^m_t))+\phi(\mu^1_x f^1_x(\pi^m_t))}{2} \right) dt \ge 0, \label{eq:ExternalGoal}
\end{align}
for the external case, and
\begin{align}
\sum_{j=1}^k  \int_{-\infty}^\infty &\sum_m \sum_{b=0}^1 \left( \phi(f_{x_j=b}(\pi^m_t)) - \frac{\phi(f_{x_j=b}^0(\pi^m_t))+\phi(f_{x_j=b}^1(\pi^m_t))}{2} \right) \nonumber \\ 
& -\sum_m \sum_x  \left( \phi(\mu_x f_x(\pi^m_t))-\frac{\phi(\mu^0_x f^0_x(\pi^m_t))+\phi(\mu^1_x f^1_x(\pi^m_t))}{2} \right) dt \ge 0, \label{eq:InternalGoal}
\end{align}
for the internal case, where
$$f_{x_j=b}(\pi^m_t):= \sum_{X: X_j=b} \mu_X f_X(\pi^m_t),$$
and 
$$\qquad f^0_{x_j=b}(\pi^m_t):= \sum_{X: X_j=b} \mu^0_X f^0_X(\pi^m_t), \qquad f^1_{x_j=b}(\pi^m_t):= \sum_{X: X_j=b} \mu^1_X f^1_X(\pi^m_t).$$

Denote the function inside the integral of (\ref{eq:InternalGoal}) by $\concav_\mu(t,j)$, and the function inside the integral of (\ref{eq:ExternalGoal}) by $\concav_\mu^\ext(t)$. Note further that by Claim~\ref{claim:No1} we can assume that $\mu_{\vec{1}}=0$.  Hence our goal reduces to show the following:

\begin{statement}[First reduction]
\label{stat:1st}
To prove Theorem~\ref{thm:mainMultiparty} it suffices to assume $\mu$ satisfies $\mu(\vec{1})=0$, and verify 
$$\int_{-\infty}^\infty \concav_\mu^\ext(t) dt \ge 0 \qquad \mbox{and} \qquad \sum_{j=1}^k \int_{-\infty}^\infty \concav_\mu(t,j) dt \ge  0.$$
\end{statement}
Recall we assumed $t_s=0$ by shifting the time. The next two claims show that one only needs to focus on the interval $[-\gamma_0,\gamma_1]$. 
\begin{myclaim}  \label{claim:less-than-ts-1-contributes-nonnegative}
We have $ \int_{-\infty}^{- \gamma_0} \concav_\mu^\ext(t) dt \ge 0$ and $\sum_{j=1}^k  \int_{-\infty}^{- \gamma_0} \concav_\mu(t,j)dt \ge 0$.
\end{myclaim} 
\begin{proof}
Observe that $\Pi, \Pi^0$ and $\Pi^1$ are identical up to time $- \gamma_0$. Let $\Pi_P$ denote a similar  protocol, with the only difference that in $\Pi_P$ at time $t=- \gamma_0$ all the players reveal their inputs. Then,
$$\int_{-\infty}^{- \gamma_0} \concav_\mu^\ext(t) dt
= \entropy(X|\Pi_P)-\entropy(X|\Pi_P, B) \ge 0,$$
and
\[   \nonumber
\sum_{j=1}^k \int_{-\infty}^{- \gamma_0} \concav_\mu(t,j)dt
= \sum_{j=1}^k (\entropy(X|X_j, \Pi_P)-\entropy(X|X_j, \Pi_P, B)) \ge 0. \tag*{\qedhere}
\]
\end{proof}

\begin{myclaim}   \label{claim:greater-than-gamma1-contributes-0}
We have $ \int_{ \gamma_1}^\infty \concav_\mu^\ext(t) dt = 0$ and $\sum_{j=1}^k  \int_{ \gamma_1}^\infty \concav_\mu(t,j)dt = 0$.
\end{myclaim} 
\begin{proof}
We use the formula in \eqref{eq:ExternalGoal} by integrating in the corresponding range $[\gamma_1, \infty)$. As $t \ge \gamma_1$, plug in $\mu_x^0 f_x^0(\pi_t^m) = (1-\epsilon \zeta_s) \mu_x f_x(\pi_t^m)$ and $\mu_x^1 f_x^1(\pi_t^m) = (1+\epsilon \zeta_s) \mu_x f_x(\pi_t^m)$, a simple calculation shows $ \int_{\gamma_1}^\infty \concav_\mu^\ext(t) dt = 0$. Similarly one can calculate the internal case.
\end{proof}

\begin{statement}[Second reduction]
\label{stat:2nd}
To prove Theorem~\ref{thm:mainMultiparty} it suffices to assume $\mu$ satisfies $\mu(\vec{1})=0$, and verify 
$$
\int_{-\gamma_0}^{\gamma_1} \concav_\mu^\ext(t) dt \ge 0 
\qquad 
\mbox{and} 
\qquad 
\sum_{j=1}^k \int_{-\gamma_0}^{\gamma_1} \concav_\mu(t,j) dt \ge  0.$$
\end{statement}

\begin{remark}
The computation result of the two-party $\AND$ done in  \cite[Section 7.7]{MR3210776} shows the concavity term (the one that we want to verify its non-negativity) can be of order $\epsilon^2$.  One will see in Section \ref{sec:multiparty-code} that our computation gives only an order $\epsilon^3$. This is because we choose to focus our computation, as allowed by Statement \ref{stat:2nd}, on the interval $[-\gamma_0, \gamma_1]$ only. Claim \ref{claim:epsilon-2-contribution} below shows an order $\epsilon^2$ term can appear if the whole range is considered. 
\end{remark}

\begin{myclaim}  \label{claim:epsilon-2-contribution}
Suppose $s \ge 2$ and $L = |t_{s-1}| > 0$. If $\gamma_0 \le L/2$, then
\[   \label{eqn:epsilon-2-bound-Ext}
\int_{t_{s-1}}^{-\gamma_0} \concav_\mu^\ext(t) dt 
\ge
\frac{(1 - e^{-(s-1)L/2}) \mu_{\vec{0}} \mu_{e_s} }{2(s-1)} \epsilon^2 \ge 0,
\]
and
\[   \label{eqn:epsilon-2-bound-Int}
\sum_{j=1}^k \int_{t_{s-1}}^{-\gamma_0} \concav_\mu(t,j) dt
\ge
\frac{(k-1)(1 - e^{-(s-1)L/2}) \mu_{\vec{0}} \mu_{e_s} }{2(s-1)} \epsilon^2 \ge 0.
\]
\end{myclaim} 

\begin{proof}
Consider external case first. Let $\mu'$ be defined as
$$
\mu'_x =
\begin{cases}
\beta \mu_x, &\quad x_s = 0, \\
-\zeta \mu_x, &\quad x_s = 1.
\end{cases}
$$
Then $\mu^0= \mu + \epsilon \mu'$ and $\mu^1= \mu - \epsilon \mu'$, hence $f^0(\pi^m_t)= f(\pi^m_t)+ \epsilon \sum \mu'_x f_x(\pi^m_t)$ and   $f^1(\pi^m_t)= f(\pi^m_t)- \epsilon \sum \mu'_x f_x(\pi^m_t)$.
Note that $f(\pi^m_t) = 0$ (in our case this happens when $m \ge s$) implies $\concav_\mu^\ext(t) = 0$. On the other hand, when $f(\pi^m_t) \neq 0$ (i.e., $1 \le m \le s-1$), using Taylor expansion at the point $f(\pi^m_t)$ for functions $\phi(f^0(\pi^m_t))$ and $\phi(f^1(\pi^m_t))$, and expansion at $\mu_x f_x(\pi^m_t)$ for functions $\phi(\mu^0_x f^0_x(\pi^m_t))$ and $\phi(\mu^1_x f^1_x(\pi^m_t))$, we obtain (note here we won't have $\epsilon^3$)
\begin{align}   \label{eqn:epsilon-square-term-Ext}
\begin{split}
\concav_\mu^\ext(t)
&\ge 
\sum_{m = 1}^{s-1} \left(\sum_x  \frac{(\mu'_x f_x(\pi^m_t))^2}{2\mu_x f_x(\pi^m_t)} - \frac{(\sum \mu'_x f_x(\pi^m_t))^2}{2f(\pi^m_t)} \right)\epsilon^2 \\
&=
\frac{1}{2} \sum_{m = 1}^{s-1} \left(\mu_{e_s} f_{e_s}(\pi^m_t) \left(1 - \frac{\mu_{e_s} f_{e_s}(\pi^m_t)}{f(\pi^m_t)} \right)  \right) \epsilon^2  \\
&\ge \frac{1}{2} \sum_{m = 1}^{s-1} \left(\mu_{e_s} f_{e_s}(\pi^m_t) \frac{\mu_{\vec{0}} f_{\vec{0}}(\pi^m_t)}{f(\pi^m_t)} \right) \epsilon^2 \ge 0.
\end{split}
\end{align}

By Statement \ref{stat:3rd} one can assume $\mu_{e_s} = \cdots = \mu_{e_k}$ and $\mu_{e_1} = \cdots = \mu_{e_{s-1}} = \mu_{e_s} e^{-L}$, thus $\mu_{\vec{0}} = 1 - (k-s+1) \mu_{e_s} - (s-1) \mu_{e_s} e^{-L}$. Then $t_{s-1} = 0$ and $t_s = L$. As $\gamma_0 \le L/2$ implies $t_s - \gamma_0 = L - \gamma_0 \ge L/2$, and \eqref{eqn:epsilon-square-term-Ext} says the integrand is non-negative, hence a lower bound is given by the integration of \eqref{eqn:epsilon-square-term-Ext} in the range $[0, L/2]$. For  $t \in [0, L/2]$ and $1 \le m \le s-1$, we have $f_{\vec{0}} (\pi_t^m) = f_{e_s} (\pi_t^m) = \ldots =f_{e_k} (\pi_t^m) = e^{-(s-1)t}$, and $f_{e_1} (\pi_t^m) = \ldots =f_{e_{s-1}} (\pi_t^m) = e^{-(s-2)t}$, thus $f(\pi_t^m) = (1 - (s-1) \mu_{e_s} e^{-L} + (s-2) \mu_{e_s} \e^{-L} e^t) e^{-(s-1)t} \le (s-1)e^{-(s-1)t}$. Hence,
\[  \nonumber
\eqref{eqn:epsilon-square-term-Ext}
\ge
\frac{s-1}{2} \mu_{e_s} f_{e_s} \frac{\mu_{\vec{0}} f_{\vec{0}} (\pi_t^m)}{(s-1)e^{-(s-1)t}} \epsilon^2
= \frac{\mu_{\vec{0}} \mu_{e_s} }{2} e^{-(s-1)t} \epsilon^2.
\]
Integrating in the range $[0, L/2]$ with respect to $t$ gives the desired bound \eqref{eqn:epsilon-2-bound-Ext}.

Similarly, in the internal case one has
\begin{align*}
\sum_{j=1}^k \concav_\mu(t,j)
&\ge
\frac{1}{2} \sum_{m = 1}^{s-1} \sum_{j=1,j\neq s}^k \left(\mu_{e_s} f_{e_s}(\pi^m_t) \frac{\mu_{\vec{0}} f_{\vec{0}}(\pi^m_t)}{f(\pi^m_t) - \mu_{e_j} f_{e_j}(\pi^m_t)} \right) \epsilon^2 \\
&\ge
\frac{k-1}{2} \sum_{m = 1}^{s-1} \left(\mu_{e_s} f_{e_s}(\pi^m_t) \frac{\mu_{\vec{0}} f_{\vec{0}}(\pi^m_t)}{f(\pi^m_t)} \right) \epsilon^2.
\end{align*}
Hence we get the bound \eqref{eqn:epsilon-2-bound-Int} after integration.
\end{proof}

\subsection{Futher reductions}    \label{sec:Further-reductions}

In this section we obtain a futher reduction of Statement~\ref{stat:2nd} that will have a constant number of variables and so one can finally verify it using Wolfram Mathematica:

\begin{statement}[Third reduction]
\label{stat:3rd}
To prove Theorem~\ref{thm:mainMultiparty} it suffices to assume $\mu$ satisfies 
\[   \label{eqn:measure-mu}
\mu_{e_1}=\dots=\mu_{e_{s-1}}=\beta, \quad
\mu_{e_{s}}= \dots =\mu_{e_k}= e^{\gamma_0} \beta, \quad
\mu_{\vec{0}} = 1-(s-1) \beta - (k-s+1) e^{\gamma_0} \beta,
\]
where $0 < \beta < 1$, and verify 
$$\int_{-\gamma_0}^{\gamma_1} \concav_\mu^\ext(t) dt \ge 0 \qquad \mbox{and} \qquad \sum_{j=1}^k \int_{-\gamma_0}^{\gamma_1} \concav_{\mu}(t,j) dt \ge  0.$$
\end{statement}

Statement~\ref{stat:3rd} follows from Claim \ref{claim:mu-k-Equals-mu-s} showing  that it suffices to consider measures $\mu$ such that $\mu_{e_j} = \mu_{e_s}$ for all $j \ge s$, together with the observation that conditioned on the buzz time $t \in [-\gamma_0, \gamma_1]$, we have $\mu_{e_1}|_{t \ge t_s - \gamma_0} = \cdots = \mu_{e_{s-1}}|_{t \ge t_s - \gamma_0}$.

\begin{myclaim}  \label{claim:SameAverage}
For every $z$, 
\begin{align*}
&\Pr\left[X=z \wedge t(\Pi_z) \in [-\gamma_0,\gamma_1] \right] \\
&=\frac{\Pr\left[X^0=z \wedge t(\Pi^0_z) \in [-\gamma_0,\gamma_1] \right]+\Pr\left[X^1=z \wedge t(\Pi^1_z) \in [-\gamma_0,\gamma_1] \right]}{2}.
\end{align*}
\end{myclaim} 

\begin{proof}
We need to show
$$\mu_z \sum_m \int_{-\gamma_0}^{\gamma_1} f_z(\pi_t^m) dt= \frac{1}{2} \left(\mu^0_z \sum_m \int_{-\gamma_0}^{\gamma_1} f_z^0(\pi_t^m) dt+ \mu_z^1 \sum_m \int_{-\gamma_0}^{\gamma_1} f_z^1(\pi_t^m) dt\right).$$

Recall that $\Phi_z(t)$ denotes the total amount of active time spent by all players before time $t$. The probability that $\Pi_z$ finishes in the interval $[-\gamma_0,\gamma_1]$  is equal to $$e^{-\Phi_z(-\gamma_0)} - e^{-\Phi_z(\gamma_1)}.$$
Denoting  by $\Phi^0_z(t)$ and $\Phi^1_z(t)$ the total active time for the protocols $\pi^\wedge_{\mu_0}$ and $\pi^\wedge_{\mu_1}$ on the input $z$, the claim is equivalent to
$$\mu_z \cdot (e^{-\Phi_z(-\gamma_0)} - e^{-\Phi_z(\gamma_1)}) = \frac{\mu^0_z \cdot (e^{-\Phi^0_z(-\gamma_0)} - e^{-\Phi^0_z(\gamma_1)})}{2} + \frac{\mu^1_z \cdot (e^{-\Phi^1_z(-\gamma_0)} - e^{-\Phi^1_Z(\gamma_1)})}{2}.$$
Since $\mu_z = \frac{\mu^0_z+\mu^1_z}{2}$ and $\Phi_z(-\gamma_0)=\Phi_z^0(-\gamma_0)=\Phi_z^1(-\gamma_0)$, the equality reduces to
$$\mu_z e^{-\Phi_z(\gamma_1)} = \frac{\mu^0_z e^{-\Phi^0_z(\gamma_1)}+\mu^1_z  e^{-\Phi^1_z(\gamma_1)}}{2}.$$ %
When $z_s=1$, $\Phi_z=\Phi_z^0=\Phi_z^1$, and thus $\mu_z = \frac{\mu^0_z+\mu^1_z}{2}$ verifies the equality. In the case of $z_s=0$, we have that $\Phi^0_z(\gamma_1)=\Phi_z(\gamma_1)+\gamma_0$, and $\Phi^1_z(\gamma_1)=\Phi_z(\gamma_1)-\gamma_1$. Substituting $\gamma_0 = \ln\left(\frac{1+\epsilon \beta_s}{1-\epsilon \zeta_s}\right)$, $\gamma_1 = \ln\left(\frac{1+\epsilon \zeta_s}{1-\epsilon \beta_s}\right)$,  $\mu^0_z=(1+\epsilon \beta_s) \mu_z$ and $\mu^1_z=(1-\epsilon \beta_s) \mu_z$ verifies the equality.
\end{proof}

\begin{myclaim}  \label{claim:mu-k-Equals-mu-s}
Suppose $\mu$ satisfies $\mu_{e_s} = \ldots = \mu_{e_{s+a}} < \mu_{e_{s+a+1}}$ for some $a \ge 0$ with $s+a+1 \le k$. Assume $\epsilon$ is sufficiently small so that $\gamma_1 \le t_{s+a+1}$. Let $\mu'$ be a measure for the $(s+a)$-player AND function, defined as: $\mu'_{\vec{0}} = \mu_0 + \sum_{j > s+a} \mu_{e_j}$, and $\mu'_{e_j} = \mu_{e_j}$ for $1 \le j \le s+a$. Then the following hold.
\begin{enumerate}[(1).]
\item $\int_{- \gamma_0}^{\gamma_1}  \concav_\mu^\ext(t) dt = \int_{- \gamma_0}^{\gamma_1} \concav_{\mu'}^\ext(t) dt$;
\item If $\sum_{j=1}^{s+a} \int_{- \gamma_0}^{\gamma_1} \concav_{\mu'}(t,j) dt  \ge 0$, then $\sum_{j=1}^{k} \int_{- \gamma_0}^{\gamma_1} \concav_{\mu}(t,j) dt  \ge 0$.
\end{enumerate}
\end{myclaim} 

This claim shows that to verify the concavity conditions \eqref{eq:ExternalGoal} and  \eqref{eq:InternalGoal} it suffices to consider only those measures satisfying $\mu_{e_j} = \mu_{e_s}$ for all $j > s$.

\begin{proof}
We confine ourselves in the interval $t \in [- \gamma_0, \gamma_1]$ throughout the proof. Let $f, f'$ and $\Pi, \Pi'$ denote the pdf and protocols for $\pi_\mu^\wedge$ and $\pi_{\mu'}^\wedge$, respectively.

\begin{enumerate}[(1).]
\item Obviously we have
\[   \label{eqn:temp1}
f'_{\vec{0}}(\pi_t^m) = f_{\vec{0}}(\pi_t^m), \quad\text{and\ }\quad f'_{e_j}(\pi_t^m) = f_{e_j}(\pi_t^m), \quad 1 \le j \le s+a.
\]
For the protocol $\pi_\mu^\wedge$, observe we have
\[   \label{eqn:temp2}
f_{\vec{0}}(\pi_t^m) = f_{e_j}(\pi_t^m), \quad j > s+a,
\]
for all $m = 1, \ldots, k$. Hence \eqref{eqn:temp1} and \eqref{eqn:temp2}  imply that $f(\pi_t^m) = f'(\pi_t^m)$. Clearly similar results hold for $\Pi^0, \Pi^1$ and $\Pi'^0, \Pi'^1$ . This imply that the first integral in \eqref{eq:ExternalGoal} does not change from $\mu$ to $\mu'$.

It remains to show that the second integral in \eqref{eq:ExternalGoal} does not change either. Expand this integral gives,
\begin{equation}    \label{eqn:temp3}
\begin{aligned}
&\int_{-\gamma_0}^{\gamma_1} \sum_X \sum_m \left( f_X(\pi_t^m) \mu_X \log(\mu_X) -\frac{\mu^0_X f^0_X(\pi_t^m) \log(\mu_X)+\mu^1_X f^1_X(\pi_t^m) \log(\mu_X)}{2} \right) dt \\
 +&\int_{-\gamma_0}^{\gamma_1} \sum_X \sum_m \Bigg(\mu_X \phi(f_X(\pi_t^m)) - \phantom{\Bigg)} \\
&\phantom{\Bigg(} \frac{\mu^0_X \left(\phi(f^0_X(\pi_t^m))+ f^0_X(\pi_t^m) \log(1+\epsilon \beta)\right) + \mu^1_X \left(\phi(f^1_X(\pi_t^m))+ f^1_X(\pi_t^m) \log(1-\epsilon \beta)\right)}{2} \Bigg) dt.
\end{aligned}
\end{equation}
By Claim~\ref{claim:SameAverage} the first integral in \eqref{eqn:temp3} is $0$. Hence it only remains to show the second integral in \eqref{eqn:temp3} does not change. But this is again a direct consequence of \eqref{eqn:temp1} and \eqref{eqn:temp2} with corresponding facts for $\Pi_X^0$ and $\Pi_X^1$.

\item By definition of the measures $\mu, \mu'$, one has
\[  \nonumber
\sum_{j=1}^{k}\int_{- \gamma_0}^{\gamma_1} \concav_{\mu}(t,j) dt 
-
\sum_{j=1}^{s+a} \int_{- \gamma_0}^{\gamma_1} \concav_{\mu'}(t,j) dt 
=
\sum_{j=s+a+1}^{k} \int_{- \gamma_0}^{\gamma_1} \concav_{\mu}(t,j) dt.
\]
Hence it suffices to show $\sum_{j=s+a+1}^{k}  \int_{- \gamma_0}^{\gamma_1} \concav_{\mu}(t,j) dt  \ge 0$.

Let $\mu_{X_j = b}$ denote the distribution $\mu$ of $X$ conditioned on $X_j = b$, one can check that
\[   \label{eqn:Connection-Ext-Int}
\int_{-\gamma_0}^{\gamma_1} \concav_{\mu}(t,j) dt =  
\Ex_{b} \int_{-\gamma_0}^{\gamma_1} \concav^\ext_{\mu_{X_j = b}} (t) dt.
\]
In section \ref{sec:ExtComputation} we show the external concavity condition indeed holds, hence \eqref{eqn:Connection-Ext-Int} implies $\int_{- \gamma_0}^{\gamma_1} \concav_{\mu}(t,j) dt  \ge 0$, as desired. \qedhere
\end{enumerate}
\end{proof}

\begin{proof}[Proof of Theorem~\ref{thm:mainMultiparty}]
We use $\wedge$ to denote the multiparty $\AND$ function.
Consider the external case first. Recall we set $\Omega$ to be the set of all external trivial measures together with the measure in Claim~\ref{claim:protocol-is-opt-when-mu-i-are-all-equal}, hence Condition (i) and (ii) are satisfied. Picking $w(x) = ck^{-20} x^4$ for some fixed constant $c > 0$, we verify Condition (iii)  in Section  \ref{sec:ExtComputation}. Hence $\IC^\ext_\mu(\pi^\wedge) \le \IC^\ext_\mu(\wedge)$, as $\IC^\ext_\mu(\pi^\wedge)$ is also an upper bound, hence we proved $\IC^\ext_\mu(\pi^\wedge) = \IC^\ext_\mu(\wedge)$. 

Similarly for the internal case the concavity Condition (iii) is verified in Section \ref{sec:IntComputation}. 
\end{proof}

\subsection{Information cost of multiparty AND function}   \label{sec:multiparty-code}
To simplify the notation, since every function has the argument $\pi_t^m$, we sometimes omit it from the writing while knowing it was there, such as we write $f$ to mean $f(\pi_t^m)$. We will use $\mu_0, \mu_j, f_0, f_j$ instead of $\mu_{\vec{0}}, \mu_{\e_{j}}, f_{\vec{0}}, f_{e_j}$ when there is no ambiguity, similar notations are used for measures $\mu^0, \mu^1$ and functions $f^0, f^1$.

\subsubsection{Taylor expansions} \label{sec:multi-party-Taylor}

Recall $\beta_s = \mu_s$ and $\zeta_s = 1 - \beta_s$. By the claims in Section \ref{sec:Further-reductions}, we can assume that
\[   \label{eqn:measure-mu}
\mu_1=\dots=\mu_{s-1}=\beta, \quad
\mu_s=\dots =\mu_k= e^{\gamma_0} \beta, \quad
\mu_0 = 1-(s-1) \beta - (k-s+1) e^{\gamma_0} \beta.
\]
Observe that $0 < \beta < 1/k$ (the measure when $\beta=0$ is both external and internal trivial). Furthermore, viewing $\gamma_0$ and $\gamma_1$ as functions of $\epsilon$, plugging $\beta_s = e^{\gamma_0} \beta$ and $\zeta_s = 1 - \beta_s = 1 - e^{\gamma_0}\beta$ into \eqref{eqn:gamma-0} and \eqref{eqn:gamma-1}, by implicit differentiation, we have
\[   \label{eqn:derivatives-of-gamma}
\begin{cases}
\gamma_0(0) = 0, \gamma_0'(0) = 1, \gamma_0''(0) = 1 - 2\beta,  \gamma_0'''(0) = 2 - 10 \beta + 8 \beta^2; \\
\gamma_1(0) = 0, \gamma_1'(0) = 1, \gamma_1''(0) = 2\beta - 1,  \gamma_1'''(0) = 2 + 6 \beta^2.
\end{cases}
\]
These derivatives are used in the Mathematica computation. To simplify the notation we let $\zeta = \zeta_s = 1 - e^{\gamma_0}\beta$.

Assuming $\epsilon<1/2$, then $\gamma_0+\gamma_1 \le 2 \ln (1+2\epsilon) \le 4 \epsilon$. Also note $|e^{-x}-1| \le x$  for $x \ge 0$, which together with the fact that $\Phi_x(t) \le k(\gamma_0+\gamma_1) \le 4 k \epsilon$ implies the following:
\begin{itemize}
\item $\mu_0 f_0(\pi^m_t)$ is either $0$ or close to $1-k\beta$ with distance bounded by $4k\epsilon$;
\item for $j\neq 0$, we have $\mu_j f_j(\pi^m_t)$ is either $0$ or close to $\beta$ with distance bounded by $4(k+1)\epsilon$;
\item $f(\pi^m_t)$ is either $0$ or close to $1-\beta$ with distance bounded by $16k^2\epsilon$;
\item for $j\neq 0$ and $j \neq m$, we have $f(\pi_t^m) - \mu_j f_j(\pi_t^m)$ is either $0$ or close to $1-2\beta$ with distance bounded by $16k^2\epsilon$.
\end{itemize}
Note that $f(\pi_t^m) = 0$ implies $\mu_x f_x(\pi_t^m) = 0$ for every $x$, hence $\phi(\cdot) = 0$ for all functions under consideration. On the other hand when $f(\pi_t^m) \neq 0$, then all $\mu_x f_x(\pi_t^m)$ are nonzero except when $x_m = 1$. In this case these functions $\phi(\cdot)$ have the following Taylor expansions at corresponding points as follows,
\begin{align*}
&\phi(\mu_0 f_0(\pi^m_t))
= -\frac{1-k\beta}{2} + (\ln(1-k\beta)) \mu_0 f_0(\pi^m_t) + \frac{(\mu_0 f_0(\pi^m_t))^2}{2(1-k\beta)} + O(\epsilon^3), \\
&\phi(\mu_j f_j(\pi^m_t))
= -\frac{\beta}{2} + (\ln\beta) \mu_j f_j(\pi^m_t) + \frac{(\mu_j f_j(\pi^m_t))^2}{2\beta} + O(\epsilon^3), \  j \neq 0, j\neq m \\
&\phi(f(\pi^m_t))
= -\frac{1-\beta}{2}  + (\ln(1-\beta)) f(\pi_t^m) + \frac{f(\pi_t^m)^2}{2(1-\beta)} +  O(\epsilon^3),  \\
&\phi(f(\pi_t^m) - \mu_j f_j(\pi_t^m))
= -\frac{1-2\beta}{2}  + (\ln(1-2\beta)) f(\pi_t^m) + \frac{f(\pi_t^m)^2}{2(1-2\beta)} - \\
&\phantom{=\ } (\ln(1-2\beta)) \mu_j f_j(\pi_t^m) + \frac{(\mu_j f_j(\pi_t^m))^2}{2(1-2\beta)} - \frac{\mu_j f_j(\pi_t^m) f(\pi_t^m)}{1-2\beta} +  O(\epsilon^3), \ j\neq 0, j\neq m.
\end{align*}
Recall Taylor's theorem with the remainder in Lagrange form says that the error term $O(\epsilon^3)$ in the expansion of $\phi(\mu_0 f_0(\pi_t^m))$ equals $\frac{|\phi^{(3)}(\xi)|}{6} |\mu_0 f_0(\pi_t^m) - (1-k\beta)|^3$ for some $\xi$ between $\mu_0 f_0(\pi_t^m)$ and $1-k\beta$. Since $|\mu_0 f_0(\pi_t^m) - (1-k\beta)| \le 4k\epsilon$, we have $|\xi - (1-k\beta)| \le 4k\epsilon$, hence $\xi \ge (1-k\beta) - 4k\epsilon > 0$ if $\epsilon < \frac{1-k\beta}{4k}$. Furthermore, we have $0 < \frac{1}{(1-k\beta) - 4k\epsilon} \le \frac{2}{1-k\beta}$ as long as $\epsilon \le \frac{1-k\beta}{8k}$. Therefore,
\begin{align*}   \nonumber
\frac{|\phi^{(3)}(\xi)|}{6} |\mu_0 f_0(\pi_t^m) - (1-k\beta)|^3
&= \frac{1}{6\xi^2} |\mu_0 f_0(\pi_t^m) - (1-k\beta)|^3  \\
&\le \frac{1}{6(1-k\beta - 4k\epsilon)^2} (4k\epsilon)^3 \\
&\le \frac{4}{6(1-k \beta)^2} (4k)^3 \epsilon^3 \le \frac{k^{11}}{6(1-k \beta)^2} \epsilon^3,
\end{align*}
when $0 < \epsilon \le \frac{1-k\beta}{k^7} < \frac{1-k\beta}{8k}$.  Denote the constant in this upper bound by $R_1$.

Similarly, let $R_2, R_3$ and $R_4$ denote the constants that we can get as upper bounds of the absolute values of error terms in the expansions of $\phi(\mu_j f_j(\pi^m_t)), \phi(f(\pi^m_t))$ and $\phi(f(\pi_t^m) - \mu_j f_j(\pi_t^m))$, respectively. We have
\[   \label{eqn:Remainders-bound}
\begin{cases}
R_1 
\le \frac{k^{11}}{6(1-k \beta)^2},  &\text{when\ } 0 < \epsilon \le \frac{1-k\beta}{k^7}; \\
R_2 
\le \frac{k^{14}}{6 \beta^2} , &\text{when\ } 0 < \epsilon \le \frac{\beta}{k^7}; \\
R_3 
\le \frac{k^{20}}{6 (1-\beta)^2}, &\text{when\ } 0 < \epsilon \le \frac{1-\beta}{k^7}; \\
R_4 
\le \frac{k^{20}}{6 (1-2\beta)^2}, &\text{when\ } 0 < \epsilon \le \frac{1-2\beta}{k^7}.
\end{cases}
\]

Observe that $\mu_x^0 f_x^0$ and $\mu_x^1 f_x^1$ are both close to $\mu_x f_x$ with distance bounded by $3\epsilon$, hence the corresponding functions $\phi(\mu_x^0 f_x^0)$ and $\phi(\mu_x^1 f_x^1)$ have the same expansions as above, the same holds for functions $\phi(f^0), \phi(f^1)$ and $\phi(f^0 - \mu^0_x f^0_x), \phi(f^1 - \mu^1_x f^1_x)$.

We continue to use the Taylor expansions to expand the concavity conditions \eqref{eq:ExternalGoal} and \eqref{eq:InternalGoal}.

\begin{itemize}
\item Taylor expansion of external concavity condition \eqref{eq:ExternalGoal}.

When $f(\pi_t^m) \neq 0$, we have the following expansion,
\begin{align}  \label{eqn:extExpand}
\begin{split}
&\phi(f(\pi_t^m)) - \sum_x \phi(\mu_x f_x(\pi_t^m)) \\
&= \phi(f(\pi_t^m)) - \phi(\mu_0 f_0(\pi_t^m)) - \sum_{j=1, j\neq m}^k \phi(\mu_j f_j(\pi_t^m)) \\
&= (\ln(1-\beta)) f + \frac{1}{2(1-\beta)} f^2  - (\ln(1-k\beta)) \mu_0 f_0  - \frac{1}{2(1-k\beta)} (\mu_0 f_0)^2 \\
&\phantom{=\ } - \ln\beta \sum_{j=1, j\neq m}^k  \mu_j f_j - \frac{1}{2\beta}  \sum_{j=1, j\neq m}^k (\mu_j f_j)^2 + O(\epsilon^3),
\end{split}
\end{align}
where constant in $O(\epsilon^3)$ can be bounded by $R_1 + (k-1)R_2 + R_3$ according to \eqref{eqn:Remainders-bound}. Using Claim \ref{claim:SameAverage}, we see that the first, third and fifth terms in \eqref{eqn:extExpand} become $0$ in \eqref{eq:ExternalGoal}. Let
\[  \nonumber
F^{\ext}_m (t) = \frac{1}{2(1-\beta)} f^2 - \frac{1}{2(1-k\beta)} (\mu_0 f_0)^2 - \frac{1}{2\beta}  \sum_{j=1, j\neq m}^k (\mu_j f_j)^2.
\]
Define $F_m^{\ext,0}(t)$ and $F_m^{\ext,1}(t)$ with $f$ replaced by $f^0$ and $f^1$, respectively, etc. Observe that $F^{\ext}_m (t) = 0$ when $f(\pi_t^m) = 0$, hence in general $F^{\ext}_m (t)$ is a correct representation of $\phi(f(\pi_t^m)) - \sum_x \phi(\mu_x f_x(\pi_t^m))$. Therefore, what we want to verify in \eqref{eq:ExternalGoal} becomes
\[  \label{eqn:ExtGoalCompute}
\int_{-\gamma_0}^{\gamma_1} \sum_{m=1}^k \left( F^{\ext}_m (t) - \frac{F_m^{\ext,0} (t) + F_m^{\ext,1} (t)}{2} \right) dt + O(\epsilon^4).
\]
As $\gamma_0 + \gamma_1 \le 4\epsilon$, by \eqref{eqn:Remainders-bound}, the constant in $O(\epsilon^4)$ in \eqref{eqn:ExtGoalCompute} can be bounded by
\[   \nonumber
8k(R_1 + (k-1)R_2 + R_3) \le 4 k^{21} \left( \frac{1}{(1-k\beta)^2} + \frac{1}{\beta^2} \right),
\]
when $\epsilon \le \frac{1}{k^7} \min\{\beta, 1-k\beta\}$.

\item Taylor expansion of internal concavity condition \eqref{eq:InternalGoal}.

A direct calculation gives,
\begin{align}  \label{eqn:Int-temp1}
\begin{split}
&\phantom{=} \sum_{j=1}^k \sum_{b=0,1} \phi(f_{x_j=b}(\pi_t^m)) - k \sum_x \phi(\mu_x f_x(\pi_t^m)) \\
&= \phi(f(\pi_t^m)) - k \phi(\mu_0 f_0(\pi_t^m)) \\
&+ \sum_{j=1, j \neq m}^k \Big( \phi(f(\pi_t^m) - \mu_j f_j(\pi_t^m)) - (k-1) \phi(\mu_j f_j(\pi_t^m)) \Big).
\end{split}
\end{align}
As did for the external case, when $f(\pi_t^m) \neq 0$ the above formula expands as follows,
\begin{align}  \label{eqn:intExpand}
\begin{split}
\eqref{eqn:Int-temp1} &= \left(\ln(1-\beta) + (k-2) \ln(1-2\beta) - (k-1) \ln\beta \right) f \\
&\phantom{=\ } + \left( \frac{1}{2(1-\beta)} + \frac{k-3}{2(1-2\beta)} \right) f^2  \\
&\phantom{=\ } + \left( \ln(1-2\beta) + (k-1)\ln\beta - k \ln(1-k\beta) \right) \mu_0 f_0 \\
&\phantom{=\ } - \frac{k}{2(1-k\beta)} (\mu_0 f_0)^2 + \left( \frac{1}{2(1-2\beta)} - \frac{k-1}{2\beta} \right) \sum_{j=1, j\neq m}^k (\mu_j f_j)^2\\
&\phantom{=\ } 
+ \frac{1}{1-2\beta} \mu_0 f_0 f + O(\epsilon^3).
\end{split}
\end{align}
Claim \ref{claim:SameAverage} implies the first and third terms in \eqref{eqn:intExpand} become $0$ in \eqref{eq:InternalGoal}. Let
\begin{align*}
F_m (t)
&= \left( \frac{1}{2(1-\beta)} + \frac{k-3}{2(1-2\beta)} \right) f^2 - \frac{k}{2(1-k\beta)} (\mu_0 f_0)^2  \\
&\phantom{=\ } + \left( \frac{1}{2(1-2\beta)} - \frac{k-1}{2\beta} \right) \sum_{j=1, j\neq m}^k (\mu_j f_j)^2
+ \frac{1}{1-2\beta} \mu_0 f_0 f.
\end{align*}
Define $F_m^0(t)$ and $F_m^1(t)$ similarly. Then $F_m (t)$ is a correct representation for \eqref{eqn:Int-temp1}. Therefore what we want to verify in \eqref{eq:InternalGoal} becomes
\[  \label{eqn:IntGoalCompute}
\int_{-\gamma_0}^{\gamma_1} \sum_{m=1}^k \left( F_m (t) - \frac{F_m^0 (t) + F_m^1 (t)}{2} \right) dt + O(\epsilon^4).
\]
\end{itemize}

\subsubsection{Density functions in explicit form}   \label{sec:multi-party-density-functions}
We continue to calculate functions explicitly that will be used for computing.

\begin{itemize}
\item In the protocol $\pi_\mu^\wedge$.

Consider the interval $t \in [-\gamma_0, 0)$. Let $A = (s-1)(t+\gamma_0) = (s-1)t + (s-1)\gamma_0$, The total active time $\Phi_0(t) = \Phi_j(t) = A$ for $s \le j \le k$, and $\Phi_j(t) = A-(t+\gamma_0)$ for $1 \le j \le s-1$. Hence for $1 \le m \le s-1$, we have,
\[  \nonumber
\mu_j f_j(\pi_t^m) =
\begin{cases}
\mu_0  e^{-A}, &\quad j=0,\\
\mu_j e^{t+\gamma_0} e^{-A} = e^{\gamma_0} \beta e^t e^{-A}, &\quad 1 \le j \le s-1 \text{\ and \ } j \neq m, \\
\mu_j e^{-A} = e^{\gamma_0} \beta e^{-A}, &\quad s \le j \le k, \\
0, &\quad j=m.
\end{cases}
\]
For $m \ge s$, we have $\mu_x f_x (\pi_t^m) = 0$ for all $x$. Therefore when $t \in [-\gamma_0, 0)$,
\[  \nonumber
f(\pi_t^m) =
\begin{cases}
0, &\quad m \ge s, \\
(1 - (s-1)\beta + (s-2)e^{\gamma_0} \beta e^t) e^{-A}, &\quad 1 \le m \le s-1.
\end{cases}
\]

Similarly for the interval $t \in [0, \gamma_1)$, let $B= (s-1)(t+\gamma_0) + (k-s+1)t = kt + (s-1)\gamma_0$,  the total active time is $\Phi_0(t) = B$, $\Phi_j(t) = B-(t+\gamma_0)$ for $1 \le j \le s-1$, and $\Phi_j(t) = B-t$ for $s \le j \le k$. Hence for all $1 \le m \le k$ we have,
\[  \nonumber
\mu_j f_j(\pi_t^m) =
\begin{cases}
\mu_0 e^{-B}, &\quad j=0, \\
\mu_j e^{t+\gamma_0} e^{-B} = e^{\gamma_0} \beta e^t e^{-B}, &\quad 1 \le j \le s-1 \text{\ and\ } j \neq m, \\
\mu_j e^t e^{-B} = e^{\gamma_0} \beta e^t e^{-B}, &\quad s \le j \le k \text{\ and\ } j \neq m, \\
0, &\quad j=m.
\end{cases}
\]
Therefore when $t \in [0, \gamma_1)$,
\[  \nonumber
f(\pi_t^m) = (1-(s-1)\beta  - (k-s+1) e^{\gamma_0} \beta + (k-1) e^{\gamma_0} \beta e^t) e^{-B}.
\]

\item In the protocol $\pi^{\wedge}_{\mu^0}$.

Using results from Section \ref{sec:distribution-mu0-mu1}, we have,
\[  \nonumber
\mu_x^0 f_x^0(\pi_t^m) =
\begin{cases}
(1-\epsilon \zeta)e^{-t} \mu_x f_x(\pi_t^m), &\quad t \in [-\gamma_0, 0),  x_s = 0, m\neq s, \\
(1-\epsilon \zeta) \mu_x f_x(\pi_t^m), &\quad t \in [-\gamma_0, 0), x_s = 1, m\neq s, \\
(1-\epsilon \zeta) \mu_x f_x(\pi_t^m), &\quad t \in [0, \gamma_1).
\end{cases}
\]
For the special case $m=s$ and $t \in [-\gamma_0, 0)$, we have,
\[  \nonumber
\mu_j^0 f_j^0(\pi_t^s) =
\begin{cases}
(1-\epsilon \zeta) \mu_0 e^{-t} e^{-A}, &\quad t \in [-\gamma_0, 0), j=0, \\
(1-\epsilon \zeta) e^{\gamma_0} \beta e^{-A}, &\quad t \in [-\gamma_0, 0), 1 \le j \le s-1, \\
0, &\quad t \in [-\gamma_0, 0), j=s\\
(1-\epsilon \zeta) e^{\gamma_0} \beta e^{-t} e^{-A}, &\quad t \in [-\gamma_0, 0), s+1 \le j \le k.
\end{cases}
\]
Therefore when $t \in [-\gamma_0, 0)$,
\[  \nonumber
f^0(\pi_t^m) =
\begin{cases}
(1-\epsilon \zeta) ((1 - e^{\gamma_0} \beta - (s-1) \beta)e^{-t} + (s-1) e^{\gamma_0} \beta) e^{-A}, &1 \le m \le s, \\
0, &s+1 \le m \le k.
\end{cases}
\]
When $t \in [0, \gamma_1)$, it is simply,
\[  \nonumber
f^0(\pi_t^m) = (1-\epsilon \zeta) f(\pi_t^m).
\]

\item In the protocol $\pi^{\wedge}_{\mu^1}$.

Using results from Section \ref{sec:distribution-mu0-mu1}, when $m=s$, then $\mu_x^1 f_x^1(\pi_t^m) = 0$ for all $x$ for $t \in [-\gamma_0, \gamma_1]$. Therefore $f^1(\pi_t^s) = 0$ for all $t \in [-\gamma_0, \gamma_1]$.

When $m \neq s$, we have,
\[  \nonumber
\mu_x^1 f_x^1(\pi_t^m) =
\begin{cases}
\mu^1_x f_x(\pi_t^m), &\quad t \in [-\gamma_0, 0), \\
(1-\epsilon e^{\gamma_0} \beta) e^t \mu_x f_x(\pi_t^m), &\quad t \in [0, \gamma_1), x_s = 0, \\
(1 + \epsilon \zeta) \mu_x f_x(\pi_t^m), &\quad t \in [0, \gamma_1), x_s = 1.
\end{cases}
\]
Hence when $t \in [-\gamma_0, 0)$, we have $f^1(\pi_t^m) = (1-\epsilon e^{\gamma_0} \beta) f(\pi_t^m) + \epsilon \mu_s f_s(\pi_t^m)$,
and when $t \in [0, \gamma_1)$ we have $f^1(\pi_t^m) = (1-\epsilon e^{\gamma_0} \beta) e^t f(\pi_t^m) + (1+\epsilon\zeta -(1-\epsilon e^{\gamma_0} \beta)e^t) \mu_s f_s(\pi_t^m)$. Plug in $f$ we get, when $t \in [-\gamma_0, 0)$,
\[  \nonumber
f^1(\pi_t^m) =
\begin{cases}
0, &m \ge s, \\
(1 + e^{\gamma_0} \beta(1-\epsilon e^{\gamma_0} \beta) ((s-2)e^t - (s-1)e^{-\gamma_0})) e^{-A}, &1 \le m \le s-1.
\end{cases}
\]
When $t \in [0, \gamma_1)$,
\[  \nonumber
f^1(\pi_t^m) =
\begin{cases}
0, &m = s, \\
(1 + e^{\gamma_0} \beta(1-\epsilon e^{\gamma_0} \beta)((k-2)e^t - (s-1)e^{-\gamma_0} - k + s)) e^t e^{-B}, &m \neq s.
\end{cases}
\]
\end{itemize}

\subsubsection{External information cost}  \label{sec:ExtComputation}
Using Wolfram Mathematica with results from Section \ref{sec:multi-party-Taylor} and \ref{sec:multi-party-density-functions}, we obtain
\[   \label{eqn:Ext-numeric}
\eqref{eqn:ExtGoalCompute} = \frac{(k+5s-6)(1-2\beta)\beta}{12(1-\beta)\ln 2} \epsilon^3 + O(\epsilon^4).
\]
Therefore, using the bound of the error term given in Section \ref{sec:multi-party-Taylor}, one finds $\eqref{eqn:Ext-numeric} > 0$ as long as
\[    \nonumber
\epsilon < \min\left\{\frac{(k+5s-6)(1-2\beta)\beta}{12(1-\beta)\ln 2} {\bigg/} 4 k^{21} \left( \frac{1}{(1-k\beta)^2} + \frac{1}{\beta^2} \right), \frac{1}{k^7} \min\{\beta, 1-k\beta\} \right\}.
\]
Note that $\frac{2}{1/x + 1/y} = \frac{2xy}{x+y} \ge \min\{x, y\}$ for all $x, y > 0$. Simplifying the above formula, one obtains $\eqref{eqn:Ext-numeric} > 0$ as long as
\[   \nonumber
\epsilon < c k^{-20} \min\{\beta, 1-k\beta\}^3,
\]
for some constant $c > 0$. So we have verified the concavity condition \eqref{eq:ExternalGoal} is satisfied for all $\epsilon$-weak signals such that $\epsilon$ is no greater than $c k^{-20} \min\{\beta, 1-k\beta\}^3$.

Let $\mu^E$ denote the distribution in Claim~\ref{claim:protocol-is-opt-when-mu-i-are-all-equal}, we have $|\mu - \mu^E| \le 1-k\beta$. Let $\mu'$ be defined as $\mu'_{s-1} = 0$, $\mu'_{s} = e^{\gamma_0} \beta + \beta$, and $\mu'_j = \mu_j$ for all other $j$, then $|\mu - \mu'| = \beta$.
Observe that $\mu'$ is external trivial, hence $\mu^E, \mu' \in \Omega$ (the $\Omega$ we chose at the beginning of Section \ref{sec:Some-reductions}). Therefore we have $\delta(\mu) \le \min\{\beta, 1-k\beta\}$. Thus as we choose $w(x) = ck^{-20} x^4$, the concavity condition \eqref{eq:ExternalGoal} is satisfied for all $w(\delta(\mu))$-weak signals because
\[   \nonumber
w(\delta(\mu)) \le ck^{-20} \min\{\beta, 1-k\beta\}^4  < c k^{-20} \min\{\beta, 1-k\beta\}^3.
\]
By Theorem~\ref{thm:LocalCharNoErrorWeak}, we have proved the protocol $\pi^\wedge$ in Figure \ref{fig:prot} is optimal for external information cost.

\subsubsection{Internal information cost}  \label{sec:IntComputation}

Similarly, using Wolfram Mathematica, we obtain
\[   \label{eqn:Int-numeric}
\eqref{eqn:IntGoalCompute} =
\begin{cases}
\frac{(k+5s-6)(1-2\beta)\beta}{12(1-\beta)\ln 2} \epsilon^3 + O(\epsilon^4), &\quad k=2, \\
\frac{(k+5s-6)((3k-2)\beta^2 - 4(k-1)\beta + k-1)\beta}{12(1-\beta)(1-2\beta)\ln 2} \epsilon^3 + O(\epsilon^4), &\quad k \ge 3.
\end{cases}
\]
As did in Section \ref{sec:ExtComputation}, one can show \eqref{eqn:Int-numeric} is positive when $\epsilon$ is sufficiently small. And furthermore one can pick an appropriate function $w$ to verify that the concavity condition \eqref{eq:InternalGoal} is satisfied for all $w(\delta(\mu))$-weak signals. Hence by Theorem~\ref{thm:LocalCharNoErrorWeak}, our protocol is optimal for internal information cost.


\bibliographystyle{amsalpha}
\bibliography{MultiAnd}

\end{document}

%% file: header.tex
\usepackage{enumerate}
\usepackage{amsmath,amssymb,amsfonts,framed}

\usepackage{amsthm}

\usepackage{thmtools}

\usepackage{color}
\usepackage{comment}
\usepackage[perpage]{footmisc}
\usepackage[title,page]{appendix}

\renewcommand{\ge}{\geqslant}
\renewcommand{\le}{\leqslant}
\renewcommand{\epsilon}{\varepsilon}
\newcommand{\eps}{\epsilon}

\newcommand{\cL}{\mathcal L}

\newcommand{\cX}{\mathcal X}
\newcommand{\cY}{\mathcal Y}
\newcommand{\cZ}{\mathcal Z}

\newcommand{\Psymb}{{\bf Pr}}

\DeclareMathOperator*{\ProbOp}{\Psymb}

\renewcommand{\Pr}{\ProbOp}

\renewcommand{\[}{\begin{equation}}
\renewcommand{\]}{\end{equation}}